\documentclass[onecolumn,draftclsnofoot]{IEEEtran}
\usepackage{amsmath,amsfonts,amssymb,amsthm}
\usepackage{algorithmic}
\usepackage{algorithm}
\usepackage{array}
\usepackage[caption=false,font=normalsize,labelfont=sf,textfont=sf]{subfig}
\usepackage{textcomp}
\usepackage{stfloats}
\usepackage{url}
\usepackage{verbatim}
\usepackage{graphicx}
\usepackage{cite}
\usepackage{color}
\usepackage{threeparttable}
\usepackage{makecell}
\usepackage{multirow}
\newtheorem{theorem}{Theorem}
\newtheorem{lemma}{Lemma}
\newtheorem{remark}{Remark}
\newtheorem{proposition}{Proposition}

\begin{document}

\title{Graph-Aware Group Testing with Locally Clustered Infections}

\author{Jianing Li, Li Chai,~\IEEEmembership{Senior Member,~IEEE}, Hailin Zhang
\thanks{This work was supported by National Science Foundation of China under grant 62550085.}
\thanks{Jianing Li, Li Chai, Hailin Zhang are with the College of Control Science and Engineering, Zhejiang University, Hangzhou 310027, China (Email:lijn202409@zju.edu.cn; chaili@zju.edu.cn; hl\_zhang@zju.edu.cn).}
}
%


\maketitle

\begin{abstract}
Group testing has been widely used to identify infected individuals with a limited number of tests, typically under the assumption of independent infections. Recent studies have exploited correlations among individuals, but often require additional information beyond the contact graph, such as community structures, interaction strengths, or detailed infection dynamics. Such information may be unavailable, incomplete or unreliable in practice. In this work, we assume that only the contact graph is known and develop a graph-aware group testing framework that exploits localized infection clustering in pooling design, fundamental limits of recovery, and decoding. Specifically, we propose an optimal transport-based pooling design that incorporates graph proximity and pooling constraints into a unified optimization framework. We prove that, under mild conditions, the proposed design eliminates uninfected individuals with higher probability than the Bernoulli pooling design, reducing the feasible search space for decoding. Then, we characterize the family of possible infected sets induced by localized infection clustering and derive necessary conditions on the number of tests required for exact recovery, revealing a lower testing requirement than that under the combinatorial prior. For decoding, we model the infection states of the population as a piecewise-constant graph signal and propose a graph total variation regularized decoder. We establish sufficient conditions for exact recovery under the Bernoulli pooling design in both noiseless and noisy settings, and show that $\mathcal{O}\left(K\log\left(n/K\right)\right)$ tests are sufficient under mild conditions in the noiseless case. Extensive simulations on synthetic and real-world networks demonstrate the effectiveness of the proposed framework and the benefit of exploiting graph-induced correlations in group testing.
\end{abstract}

\begin{IEEEkeywords}
Non-adaptive group testing, localized infection clustering, optimal transport, sparse recovery, graph total variation.
\end{IEEEkeywords}

\section{Introduction}
Recurrent outbreaks of respiratory diseases over the past two decades, such as SARS, MERS, and COVID-19, have highlighted the critical role of large-scale screening in epidemic control. Early identification of infected individuals enables timely isolation, which is one of the most effective measures to curb the disease spread. However, screening the entire population through individual testing is impractical due to limited testing capacity, reagent shortages, and high financial costs~\cite{cheng2020diagnostic, cleary2021using}. These challenges are especially severe in the early stage of an outbreak, when the disease prevalence is very low and only a few individuals are infected. In such scenarios, it is crucial to develop testing strategies that can reliably identify all infections using as few tests as possible.

Group testing, originally proposed by Dorfman in 1943~\cite{dorfman1943detection}, is an effective approach to improve testing efficiency of large-scale screening. It combines samples from multiple individuals into pools and tests each pool. In Dorfman-style (adaptive) group testing~\cite{dorfman1943detection, hwang1972method,westreich2008optimizing, bilder2020tests,da2022simulation}, a negative test outcome declares all samples in the pool negative, whereas samples in positive pools must be tested in subsequent rounds. This strategy enables the screening of a large population with substantially fewer tests than individual testing, however, it requires multiple rounds of testing, resulting in longer turnaround time and increased operational complexity~\cite{da2022simulation}. 

In non-adaptive group testing~\cite{chan2011non,malioutov2012boolean, aldridge2014group,chan2014non,scarlett2020noisy}, all pools are designed and tested in a single round. Samples are assigned to pools according to a pooling matrix $M\in\left\{0,1\right\}^{m\times n}$, where $m$ and $n$ denote the number of tests and individuals, respectively. Each entry $M_{i,j}=1$ indicates that sample $j$ is included in the pool $i$. Commonly used random pooling designs~\cite{9905631} include: (i) the Bernoulli design, where each sample is independently included in each pool with probability $p$, (ii) the constant column weight design, where each sample is assigned to a fixed number of pools selected uniformly at random, (iii) the doubly-regular design, which constrains both the pool size and the number of pools per sample. Random pooling designs are simple to construct and analyze, but they do not utilize the correlations among individuals. Deterministic pooling designs can be constructed from Kirkman triple systems~\cite{9416868,goenka2021contact}, and Reed-Solomon error correcting code~\cite{sciadvabc5961}. These designs exploit combinatorial or algebraic structures to provide deterministic recovery guarantees. However, constructing deterministic pooling matrices may incur high computational costs for large problem sizes~\cite{bui2019efficient}, and some constructions are restricted to specific matrix dimensions~\cite{9416868}. 

Given the pooling matrix and test outcomes, a decoding algorithm is used to identify infected individuals, which can be formulated as a sparse recovery problem~\cite{malioutov2012boolean}:
\begin{equation}
  \begin{gathered}
    \min \left\|x\right\|_0 \\
    \mathrm{s.t.} \ y_i = \bigvee_{j=1}^{n}{M_{i,j}x_j},\ i=1,\ldots,m, \\
  \end{gathered}
\end{equation}
where $\bigvee$ is the boolean OR operation, $y\in\left\{0,1\right\}^m$ denotes the test outcomes, and $x\in\left\{0,1\right\}^n$ denotes the infection state vector, with $x_i=1$ if individual $i$ is infected and $x_i=0$ otherwise. In traditional group testing, the infection state vector $x$ is typically assumed to follow either the combinatorial prior, where the number of infections $K$ is fixed ($K\ll n$) and the infected set is uniformly distributed over all $K$-subsets of $n$ individuals, or the i.i.d. prior, where each individual is independently infected with a given probability\cite{10002300}. However, these priors fail to capture dependencies among individuals. During infection transmission, disease spreads through interactions between individuals, resulting in correlated infection states. 

Accordingly, recent studies have exploited correlations among individuals to improve performance of group testing. Silva et al.~\cite{silva2021group} formulated the two-stage group testing as a graph partitioning problem and showed that grouping correlated infections into relatively small pools reduces the expected number of tests. Ahn et al.~\cite{ahn2023adaptive} studied adaptive group testing under the stochastic block infection model, where transmissions occur more frequently within communities than across communities. Nikolopoulos et al.~\cite{nikolopoulos2023community} modeled the population as a collection of possibly overlapping communities and the infection probability of an individual depends on the corresponding communities. Given this model, they derived lower bounds on the number of tests for zero-error identification. They further proposed community-aware (non)-adaptive test designs, together with a loopy belief propagation based decoder. Building upon~\cite{nikolopoulos2023community}, Goenka et al.~\cite{goenka2021contact} used side information obtained from contact tracing in non-adaptive group testing. 

Beyond community structures, several studies have exploited correlations induced by infection propagation. Arasli et al.~\cite{arasli2023group} modeled disease spread from a randomly selected patient zero over a random connected graph, such that vertices within each infection cluster share the same state. Their pooling and decoding are based on a candidate cluster formation, which may be different from the true cluster formation and introduce a trade-off between the number of tests and false classifications. Nikpey et al.~\cite{10716002} considered an edge-faulty random graph, where each edge is independently retained with a given probability and vertices in each connected component are assumed to have the same state. They established upper and lower bounds on the number of tests required for different types of graphs. Subsequent work~\cite{10619244} modeled arbitrary correlations among individuals using a hypergraph with a probability distribution over its edges, and developed a greedy adaptive algorithm to update the posterior distribution on hyperedges given the previous test outcomes.

In addition, another line of works considered settings where the space of possible infected sets is restricted. Gonen et al.~\cite{9834789} represented the candidate infected sets as hyperedges of a hypergraph and studied exact recovery under the adaptive and non-adaptive group testing. Later, Lau et al.~\cite{10002300} assumed that the infected set is uniform over a subset of all size-$K$ sets and characterized information-theoretic limits for approximate recovery. They also introduced an Ising model, which uses the contact graph and edge strengths to characterize pairwise dependencies, and developed the corresponding decoders. 

The above studies exploit correlations among individuals by incorporating prior information from the contact network or infection process and demonstrate their benefits for group testing. In particular, these correlations are characterized through additional information beyond the contact graph, such as community structures~\cite{ahn2023adaptive,nikolopoulos2023community}, interaction strengths~\cite{10002300}, or detailed infection dynamics~\cite{goenka2021contact,silva2021group}. However, such information may be unavailable or unreliable in practice. For example, some networks do not exhibit well-defined community structures, limiting the applicability of community-based methods.

In our work, we assume that only the contact graph is known. The main challenge is to characterize correlations among individuals only from the graph structure, without relying on any additional information. To this end, we exploit the localized clustering of infections (see Fig.~\ref{fig:infection cluster}), where infected individuals tend to form clusters on the contact graph and neighboring individuals are more likely to share the same infection state. This observation motivates the graph-aware pooling design. Random pooling designs are constructed independently of the contact graph and may distribute individuals from the same infection cluster across multiple pools, resulting in a large number of positive pools and making it difficult to eliminate uninfected individuals. Existing graph-based designs exploit the graph structure to improve performance, but they rely on community structures~\cite{nikolopoulos2023community} or are tailored to two-stage testing~\cite{silva2021group}. In contrast, we develop a non-adaptive graph-aware pooling design using only the contact graph. Optimal transport, originated from the classical transportation problem~\cite{monge1781memoire,kantorovich1942transfer,peyre2019computational}, seeks a transport plan that minimizes the total cost while satisfying prescribed marginal constraints. This cost-based assignment mechanism is well suited for graph-aware pooling, as the transport cost can encode graph proximity and the marginal constraints can enforce the pooling constraints. Specifically, we formulate the assignment of individuals to pools as an optimal transport problem, where the transport costs are defined as shortest-path distances between individuals and the marginal constraints prescribe the pool size and the number of tests per individual. This formulation incorporates graph proximity and pooling constraints, promoting nearby individuals with correlated states to be pooled together.

The clustering of infections also provides a structural prior for recovery. We represent the infection states of the population as a graph signal, which is sparse on the vertex domain and exhibits a piecewise-constant structure over the graph. Such a structure is naturally captured by the graph total variation, which measures the sparsity of the first-order differences of a graph signal and promotes piecewise-constant solutions.

The main contributions of our work are summarized as follows:
\begin{itemize}
\item{We develop a graph-aware pooling design based on optimal transport, which incorporates graph proximity and pooling constraints into a unified optimization framework. The proposed design favors assigning nearby individuals to common pools, enabling more effective elimination of uninfected individuals and reducing the feasible search space for decoding. We further derive conditions under which it eliminates uninfected individuals with higher probability than the Bernoulli pooling design.}
\item {We characterize the family of possible infected sets induced by localized infection clustering. Different from the combinatorial prior, which assumes that the infected set is uniformly distributed over all $K$-subsets of $n$ individuals, our model restricts the candidate infected sets to a subset of all size-$K$ sets through the underlying graph structure. Then, we derive necessary conditions on the number of tests required for exact recovery, which is lower than that under the combinatorial prior.}
\item{We model the infection states of the population as a piecewise-constant graph signal. This representation captures both the sparsity of infections and the dependencies between neighboring individuals. We propose a graph total variation regularized decoder and establish sufficient conditions for exact recovery in both noiseless and noisy settings. The results show that $\mathcal{O}\left(K\log \left(n/K\right)\right)$ tests are sufficient under mild conditions in the noiseless case, which is of the same order as the information-theoretic lower bound under the combinatorial prior.}
\item {Extensive simulations on synthetic and real-world networks evaluate the proposed pooling design and decoder, showing that our graph-aware framework achieves better recovery performance than the existing group testing methods.}
\end{itemize}
\noindent\textit{Notations:} The set of real number is denoted by $\mathbb{R}$. For a number $a\in\mathbb{R}$, $\left(a\right)_{+}=\max\left\{a,0\right\}$ denotes the positive part of $a$. Sets are denoted by calligraphic letters, e.g., $\mathcal{X}$, with $\left|\mathcal{X}\right|$ representing its cardinality. For a vector $x$, $\operatorname{supp}\left(x\right)$ denotes the set of indices corresponding to the non-zero elements of $x$, and 
$\left\lVert x \right\rVert_0$ represents the $\ell_0$-norm, i.e., the number of non-zero elements in $x$. $\mathbf{1}_{n}$ denotes the all-one vector in $\mathbb{R}^n$ and $\mathbb{I}\left(\cdot\right)$ is the indicator function, which equals to $1$ if the condition inside is satisfied and $0$ otherwise. 
\begin{figure}[!t]
\centering
\includegraphics[width=0.5\linewidth]{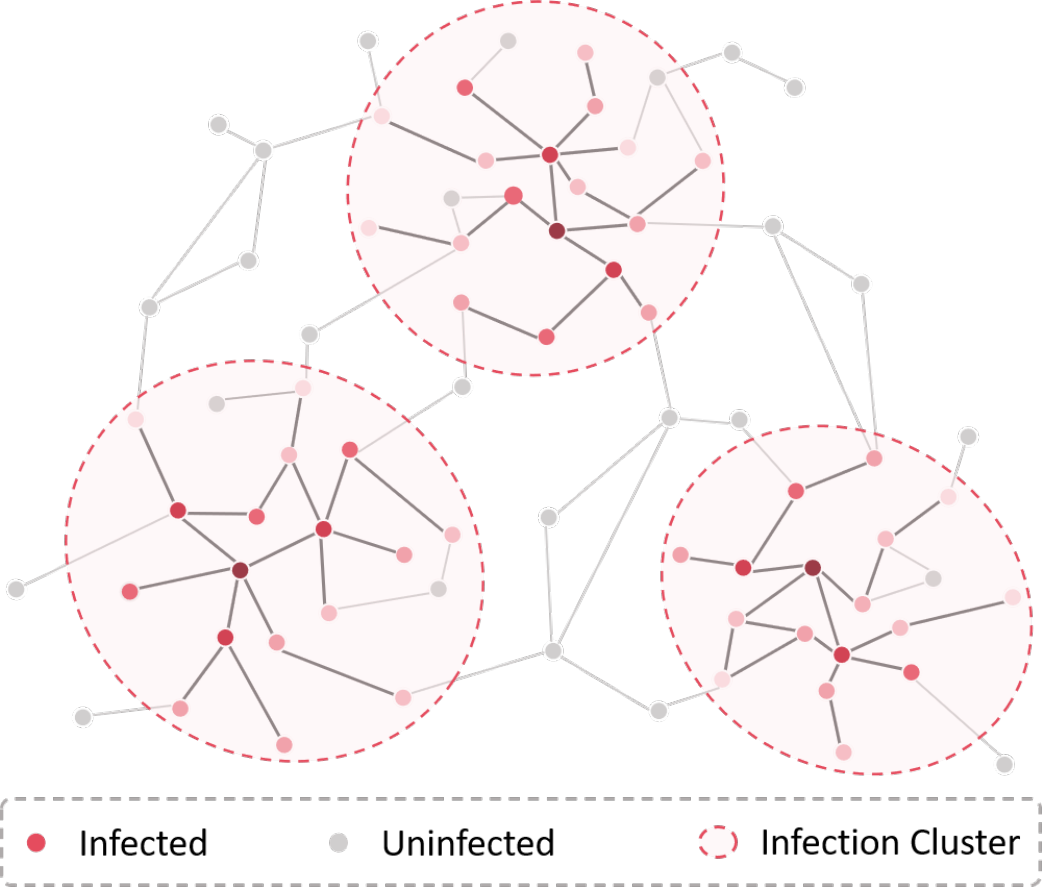}
\caption{Illustration of infection cluster. Infection sources propagate through the network and form spatially localized clusters, while uninfected individuals remain distributed outside these clusters.}
\label{fig:infection cluster}
\end{figure}

\section{Problem Formulation}\label{sec:problem formulation}
\subsection{Network and Infection Spread Model}
Let $\mathcal{G}=\left(\mathcal{V},\mathcal{E}, A\right)$ be a connected, undirected graph, where $\mathcal{V}=\{1,\ldots,n\}$ is the set of vertices and $\mathcal{E}$ is the set of edges. Each vertex represents an individual, and each edge $\left(i,j\right)\in\mathcal{E}$ indicates a contact between individuals $i$ and $j$ through which infection can be transmitted. The adjacency matrix of $\mathcal{G}$ is represented by $A\in \{0,1\}^{n\times n}$, where $A_{i,j}=1$ if $\left(i,j\right)\in\mathcal{E}$. The maximum degree of $\mathcal{G}$ is denoted by $d_{max}$, and the diameter of $\mathcal{G}$ is defined as
\begin{equation}
  D=\max_{j,j'\in\mathcal{V}}{d_{\mathcal{G}}\left(j,j'\right)},
\end{equation}
where $d_{\mathcal{G}}(j,j')$ denotes the shortest-path distance between vertices $j$ and $j'$.

We consider an infection process on the graph $\mathcal{G}$, following the susceptible-infected (SI) model~\cite{6413787}. Suppose that $k$ individuals are initially infected ($k\ll n$). The infection propagates along the edges of $\mathcal{G}$, and each infected individual transmits the infection to its susceptible neighbors independently with a certain probability. The propagation terminates when the number of infected individuals reaches $K$, yielding an infected set $\mathcal{I}\subseteq \mathcal{V}$ with $\left|\mathcal{I}\right|=K$, where $k\leq K \ll n$. Since the infection can only spread along edges of $\mathcal{G}$, the induced subgraph $\mathcal{G}\left[\mathcal{I}\right]$ has at most $k$ connected components, where $\mathcal{G}\left[\mathcal{I}\right]$ denotes the subgraph of $\mathcal{G}$ induced by the vertex set $\mathcal{I}$. Thus, the family of possible infected sets can be defined as
\begin{equation}
  \mathcal{F}_{K,k} = \left\{\mathcal{I}\subseteq \mathcal{V}: \left|\mathcal{I}\right|=K,c\left(\mathcal{G}\left[\mathcal{I}\right]\right)\leq k \right\},
\end{equation}
where $c\left(\mathcal{G}\left[\mathcal{I}\right]\right)$ denotes the number of connected components of the induced subgraph $\mathcal{G}\left[\mathcal{I}\right]$. Throughout this paper, we assume that the infected set $\mathcal{I}$ is uniformly distributed over $\mathcal{F}_{K,k}$.

We represent the infection states of the population as a graph signal $x=\left[x_1,\ldots,x_n\right]^T \in \{0,1\}^n$ defined on $\mathcal{G}$, where $x_i=1$ if individual $i$ is infected and $x_i=0$ otherwise. Thus, the infected set is given by $\mathcal{I}=\operatorname{supp}\left(x\right)$ and the vector $x$ is $K$-sparse. In addition, due to the localized spatial clustering of infections, the graph signal $x$ exhibits a piecewise-constant structure, with variations concentrated along the boundaries between infected and uninfected regions. 

\subsection{Group Testing Model}
In non-adaptive group testing, the pooling matrix is designed prior to testing. We use a binary matrix $M \in \left\{0,1\right\}^{m\times n}$ to represent the pooling process, where $m$ is the number of tests and $n$ is the number of individuals. Each row of $M$ corresponds to a pool, where $M_{i,j}=1$ indicates that the sample of individual $j$ is included in the pool $i$ and $M_{i,j}=0$ otherwise. In the noiseless setting, the test outcome vector $y=\left[y_1,\ldots,y_m\right]^T\in\left\{0,1\right\}^m$ is given by
\begin{equation}\label{eq:system_noiseless}
y_i = \bigvee_{j=1}^{n}{M_{i,j}x_j},\ i=1,\ldots,m,
\end{equation}
where $\bigvee$ is the boolean OR operation. The pool $i$ is positive if and only if it contains at least one infected individual. In the noisy setting, each test outcome is independently flipped with probability $\epsilon\in(0,1/2)$. The noisy test outcome vector $\tilde{y}=\left[\tilde{y}_1,\ldots,\tilde{y}_m\right]^T\in\left\{0,1\right\}^m$ is given by
\begin{equation}\label{eq:system_noisy}
\tilde{y}_i = y_i \oplus e_i,\ e_i\sim Bernoulli(\epsilon),\ i=1,\ldots,m,
\end{equation}
where $\oplus$ is the boolean XOR operation and $\left\{e_i\right\}_{i=1}^{m}$ are independent Bernoulli random variables.

\subsection{Problem Statement}
Given the contact graph and the infection spread model described above, our goal is two-fold: (i) design a non-adaptive pooling matrix $M$ that exploits the graph structure, and (ii) develop a decoder that recovers the infected set $\mathcal{I}$, or equivalently, the infection sate vector $x$, from the test outcomes $y$ (or $\tilde{y}$ in the noisy setting). Traditional group testing is typically studied under the combinatorial prior, where every $K$-subset of $\mathcal{V}$ is assumed to be equally likely, yielding $\binom{n}{K}$ possible infected sets. In our setting, however, infection propagates along the edges of $\mathcal{G}$ from $k$ sources, forming localized clusters and restricting the infected set to $\mathcal{F}_{K,k}$. This clustered structure induces correlations between the infection states of neighboring vertices, i.e., $x_j$ and $x_{j'}$ for $(j,j')\in\mathcal{E}$. We exploit these graph-induced correlations in both pooling design and decoding to achieve exact recovery with as few tests as possible.

\section{Optimal Transport-Based Pooling Design}\label{sec:optimal tansport-based pooling design}
In this section, we develop a graph-aware pooling design that assigns samples to pools according to their proximity on the contact graph $\mathcal{G}$. The key idea is to associate each pool with a representative vertex and assign samples closer to the representative a higher probability of being included in that pool. Accordingly, we formulate the assignment of samples to pools as an optimal transport problem, where the transport plan specifies the assignment probabilities and the transport cost is defined as the graph distance between each sample and the representative vertex of each pool. Minimizing the total transport cost therefore promotes the assignment of nearby samples to the same pools. The detailed construction is presented as follows.

We first select a set of representative vertices $\mathcal{C}=\left\{c_1,\ldots,c_m\right\} \subseteq \mathcal{V}$ via farthest-point sampling to ensure a uniform coverage of the graph, where $c_i$ is associated with the pool $i$. The first vertex $c_1$ is chosen as the graph center, i.e., 
\begin{equation}
  c_1 = \arg \min_{j\in \mathcal{V}}\max_{j'\in \mathcal{V}} d_{\mathcal{G}}(j,j').
\end{equation}
Given the set $\mathcal{C}_{l}=\left\{c_1,\ldots,c_l\right\}$, the next vertex is chosen as
\begin{equation}
  c_{l+1} = \arg \max_{j\in \mathcal{V}\backslash \mathcal{C}_l} \min_{j'\in\mathcal{C}_l} d_{\mathcal{G}}\left(j,j'\right),
\end{equation}
which ensures that each selected vertex is as far as possible from the previously selected vertices. 

Next, we define the cost matrix $C\in\mathbb{R}^{m\times n}$ with $C_{i,j}=d_{\mathcal{G}}\left(c_i,j\right)$ and formulate the transport plan as a probability matrix $P\in\left(0,1\right)^{m\times n}$, where $P_{i,j}$ denotes the probability that sample (or individual) $j$ is assigned to pool $i$. To ensure a fair comparison with the Bernoulli pooling design, we impose
\begin{equation}
  \begin{aligned}
    P\mathbf{1}_{n} = np\mathbf{1}_{m},\ P^T\mathbf{1}_{m} = mp\mathbf{1}_{n},
  \end{aligned}
\end{equation}
which prescribe the expected pool size and the expected number of pools per sample. Other constraints can also be incorporates as needed. The transport plan $P$ is obtained by solving the entropic regularized optimal transport problem:
\begin{equation}\label{eq:optimal transport}
  \begin{gathered}
    \min_{P\in\left(0,1\right)^{m\times n}} \langle C, P\rangle + \varepsilon \sum_{i,j}{P_{i,j}\left(\log\left(P_{i,j}\right)-1\right)}\\
    \mathrm{s.t.} \ P\mathbf{1}_{n} = np\mathbf{1}_{m}, \ P^T\mathbf{1}_{m} = mp\mathbf{1}_{n}, \\
  \end{gathered}
\end{equation}
where $\langle C, P\rangle$ represents the Frobenius inner product and $\varepsilon>0$. A smaller $\varepsilon$ encourages nearby individuals to be pooled together, while a larger $\varepsilon$ leads to a more random pooling design. Thus, the choice of $\varepsilon$ should be adapted to the structural characteristics of the graph.

To solve~(\ref{eq:optimal transport}), we introduce two dual variables $f=\left[f_1,\ldots,f_n\right]^T\in\mathbb{R}^n$ and $g=\left[g_1,\ldots,g_m\right]^T\in\mathbb{R}^m$ for each marginal constraint. The Lagrangian of~(\ref{eq:optimal transport}) is as follows:
\begin{equation}\label{eq:lagrange}
  \begin{aligned}
    L(P, f, g) = \langle C, P\rangle + \varepsilon \sum_{i,j}{P_{i,j}\left(\log\left(P_{i,j}\right)-1\right)} - \langle f, P^T\mathbf{1}_{m} - mp\mathbf{1}_{n}\rangle - \langle g, P\mathbf{1}_{n} - np\mathbf{1}_{m}\rangle.
  \end{aligned}
\end{equation}
First order conditions then yield
\begin{equation}
  \frac{\partial {L(P, f, g)}}{\partial {P_{i,j}}}=C_{i,j} + \varepsilon \log\left(P_{i,j}\right)-f_{j}-g_{i} = 0.
\end{equation}
Thus, the optimal solution $P^{\ast}$ is given by
\begin{equation}
  P_{i,j}^{\ast} = \exp \left(\frac{f_j+g_i-C_{i,j}}{\varepsilon}\right).
\end{equation}
Define $u=\left[u_1,\ldots,u_m\right]^T$ and $v=\left[v_1,\ldots,v_n\right]^T$, where 
\begin{equation}
  \begin{aligned}
    & u_i=\exp\left(\frac{g_i}{\varepsilon}\right),\ i=1,\ldots,m, \\
    & v_j=\exp\left(\frac{f_j}{\varepsilon}\right),\ j=1,\ldots,n.
  \end{aligned}
\end{equation}
The optimal solution $P^{\ast}$ can be rewritten as
\begin{equation}\label{eq:P_form}
  P^{\ast}_{i,j} = u_i v_j\exp\left(-\frac{C_{i,j}}{\varepsilon}\right).
\end{equation}
Then, the variables $u$ and $v$ must satisfy the following equations:
\begin{equation}\label{eq:variables_constraints}
  \begin{gathered}
    u_i \sum_{j=1}^{n}{v_j\exp\left(-\frac{C_{i,j}}{\varepsilon}\right)} = np, \\
    v_j \sum_{i=1}^{m}{u_i\exp\left(-\frac{C_{i,j}}{\varepsilon}\right)} = mp,
  \end{gathered}
\end{equation}
which can be solved via the Sinkhorn's algorithm~\cite{peyre2019computational}.

The pooling matrix $M\in\left\{0,1\right\}^{m\times n}$ is generated with $M_{i,j}\sim Bernoulli\left(P_{i,j}^{\ast}\right)$, where $P_{i,j}^{\ast}$ is the probability that sample $j$ is assigned to pool $i$. Since $P_{i,j}^{\ast}$ decreases as $C_{i,j}$, samples closer to a representative vertex have higher probabilities of being included in the corresponding pool.

The proposed graph-aware pooling design offers performance advantages over graph-agnostic random designs. In random pooling designs, samples are assigned to pools independently, which may distribute individuals from the same infection cluster across multiple pools. This dispersion increases the number of positive pools and makes it difficult to eliminate negative samples. In contrast, the proposed design favors assigning nearby samples to common pools, thereby concentrating each infection cluster within fewer pools. This results in a larger number of negative pools, enabling more effective elimination of negative samples. The theoretical analysis of the performance advantage is provided in Section~\ref{sec:decoding with graph total variation}.

\section{Information-Theoretic Lower Bound for Exact Recovery}\label{sec:information-theoretic lower bound for exact recovery}
In this section, we derive necessary conditions on the number of tests required for exact recovery under the graph-based prior. Here, the graph-based prior captures the localized clustering of infections over the graph, thereby restricting the candidate infected sets to $\mathcal{F}_{K,k}$.

We first establish upper and lower bounds on the cardinality of $\mathcal{F}_{K,k}$.

\begin{proposition}\label{proposition}
  Let $\mathcal{G}=\left(\mathcal{V},\mathcal{E}\right)$ be a connected, undirected graph with $n$ vertices and maximum degree $d_{max}$. Suppose that $k\leq \min\left\{\kappa, K\right\}$, where 
  \begin{equation}
    \begin{aligned}
      \kappa = \frac{\sqrt{\left(n+K+ed_{max}\right)^2+4nK\left(ed_{max}-1\right)}}{2\left(ed_{max}-1\right)} 
      - \frac{n+K+2-ed_{max}}{2\left(ed_{max}-1\right)}.
    \end{aligned}
  \end{equation}
  Then, we have 
  \begin{equation}
    \left|\mathcal{F}_{K,k}\right|\leq k\binom{n}{k}\binom{K-1}{k-1}(ed_{max})^{K-k} < \binom{n}{K},
  \end{equation}
  where the second inequality holds provided that
  \begin{equation}\label{eq:n_condition_necessary}
    n > K-1 + ed_{max}K\left(kK^{k-1}\right)^{\frac{1}{K-k}}.
  \end{equation}
  Moreover, 
  \begin{equation}
    |\mathcal{F}_{K,k}|\geq\max\left\{n-K+1, \left(\frac{n}{K}\right)^k\right\}.
  \end{equation}
\end{proposition}

\begin{proof}
  We first derive the upper bound on $\left|\mathcal{F}_{K,k}\right|$. For any $\mathcal{I}\in\mathcal{F}_{K,k}$, the induced subgraph $\mathcal{G}\left[\mathcal{I}\right]$ has at most $k$ connected components. Suppose that $\mathcal{G}\left[\mathcal{I}\right]$ has exactly $r$ connected components ($1\leq r\leq k$), with component sizes $s_1,\ldots,s_r$ satisfying $s_i \geq 1$ and $\sum_{i=1}^r s_i =K$. For a fixed $r$, there are at most $\binom{n}{r}$ ways to select one root for each component and $\binom{K-1}{r-1}$ ways to partition $K$ vertices into $r$ components. By~\cite[Lemma 2.3]{patel2019computing}, the number of connected induced subgraphs of size $s_i$ containing a prescribed root is at most $\left(ed_{max}\right)^{s_i-1} /2$. Thus, the number of infected sets satisfying $c\left(\mathcal{G}\left[\mathcal{I}\right]\right)=r$ is upper bounded by  
  \begin{equation}
    T_r = \binom{n}{r}\binom{K-1}{r-1}\left(ed_{max}\right)^{K-r}.
  \end{equation}
  Summing over all possible values of $r$, yields
  \begin{equation}\label{eq:F_lower_bound}
    \left|\mathcal{F}_{K,k}\right| \leq  \sum_{r=1}^{k}{T_r} = \sum_{r=1}^{k}{\binom{n}{r}\binom{K-1}{r-1}\left(ed_{max}\right)^{K-r}}.
  \end{equation}
  Under the assumption of $k\leq \kappa$, we have
  \begin{equation}
    (n-k+1)(K-k+1)\geq ed_{max}k(k-1),
  \end{equation}
  which implies that $T_r$ is monotonically increasing with respect to $r$. Hence,
  \begin{equation}
    \left|\mathcal{F}_{K,k}\right| \leq k\binom{n}{k}\binom{K-1}{k-1}(ed_{max})^{K-k},
  \end{equation}
  which is strictly smaller than $\binom{n}{K}$ provided that
  \begin{equation}
    n > K-1 + ed_{max}K\left(kK^{k-1}\right)^{\frac{1}{K-k}}.
  \end{equation}
  
  We next derive the lower bound on $\left|\mathcal{F}_{K,k}\right|$. Since every connected graph of size $n$ has at least $n-K+1$ connected induced subgraphs of size $K$~\cite[Theorem 4]{dossou2018graphs}, it follows that
  \begin{equation}
    \left|\mathcal{F}_{K,k}\right|\geq n-K+1.
  \end{equation}
  In addition, due to the connectivity of the graph, there are $\binom{n}{k}$ possible source sets, each of which produces at least one infected set in $\mathcal{F}_{K,k}$, while each infected set contains at most $\binom{K}{k}$ source sets. Hence,
  \begin{equation}
    \left|\mathcal{F}_{K,k}\right| \binom{K}{k} \geq \binom{n}{k},
  \end{equation} 
  which yields 
  \begin{equation}
    \left|\mathcal{F}_{K,k}\right| \geq \left(\frac{n}{K}\right)^k.
  \end{equation}
  This completes the proof.
\end{proof}

\begin{remark}
  Since the number of infections $K$ is usually larger than the number of sources $k$, i.e., $K\gg k$, the term $\left(kK^{k-1}\right)^{\frac{1}{K-k}}$ is close to one. Hence, the condition in~(\ref{eq:n_condition_necessary}) can be approximated by
  \begin{equation}
    n > (1+ed_{max})K,
  \end{equation}
  which is mild in the regime $k < K \ll n$.
\end{remark}

We now establish necessary conditions for exact recovery. Define the error probability as $\mathbb{P}\left(\hat{\mathcal{I}}\neq \mathcal{I}\right)$, where $\hat{\mathcal{I}}$ is an estimate of $\mathcal{I}$.

\begin{theorem}[Necessary Conditions for Exact Recovery]\label{thm:necessary conditions}
In the noiseless setting, any non-adaptive group testing algorithm that has an error probability of at most $\delta$ requires
\begin{equation}
  m \geq \left(1-\delta\right)\log_2 {\left|\mathcal{F}_{K,k}\right|}-1.
\end{equation}
Moreover, in the noisy setting with flip probability $\epsilon\in\left(0,1/2\right)$, it requires
\begin{equation}
  m \geq \frac{\left(1-\delta\right)\log_2 {\left|\mathcal{F}_{K,k}\right|}-1}{1-H_b(\epsilon)},
\end{equation}
where $H_b(\cdot)$ denotes the binary entropy function.
\end{theorem}

\begin{proof}
  The proof is analogous to~\cite[Theorems 1 and 2]{chan2011non}.
\end{proof}

\begin{remark}
  Under the combinatorial prior, the number of infected sets is $\binom{n}{K}$ and the information-theoretic lower bound scales as $K\log_2\left(n/K\right)$. In contrast, under the graph-based prior, the infected set is restricted to $\mathcal{F}_{K,k}$, whose cardinality is smaller than $\binom{n}{K}$ by Proposition~\ref{proposition}. Consequently, the information-theoretic lower bound is determined by $\log_2 \left|\mathcal{F}_{K,k}\right|$, which scales approximately as  
  \begin{equation}
    k\log_2\left(\frac{nK}{k^2 d_{max}}\right) + K\log_2 (d_{max})
  \end{equation}
  by ignoring lower order terms.
\end{remark}

\section{Decoding with Graph Total Variation}\label{sec:decoding with graph total variation}
In this section, we develop a graph total variation regularized decoder to recover the infected set $\mathcal{I}$ (or the graph signal $x$) under the graph-based prior. Theorem~\ref{thm:necessary conditions}
characterizes necessary conditions on the number of tests for any method to achieve exact recovery with high probability, whereas here we derive sufficient conditions under which the proposed decoder achieves exact recovery in both noiseless and noisy settings. We further analyze the graph-aware pooling design and demonstrate its advantage for decoding.

\subsection{Graph Total Variation Regularized Decoder}
The infected set $\mathcal{I}$, i.e., the support of graph signal $x\in\left\{0,1\right\}^n$, is generated by multi-source propagation over the graph $\mathcal{G}$. Due to the localized clustering of infections, the graph signal $x$ exhibits a piecewise-constant structure. To capture this structural property, we introduce graph total variation, which is defined as
\begin{equation}\label{eq:graph_total_variation}
  \left\|Bx\right\|_1  = \sum_{\left(i,j\right)\in\mathcal{E}} {\left|x_i-x_j\right|},
\end{equation}
where $B$ is a $\left|\mathcal{E}\right|\times \left|\mathcal{V}\right|$ graph incidence matrix. Since the $\ell_1$-norm in (\ref{eq:graph_total_variation}) promotes sparsity of the first-order differences of a graph signal, graph total variation regularization encourages piecewise-constant graph signals~\cite{7117446,wang2016trend,8926407,chen2021graph}. 

To recover $x$, we introduce the graph total variation regularizer into the sparse recovery formulation, yielding the following optimization problem:
\begin{equation}\label{eq:recovery_problem_noiseless}
  \begin{gathered}
    \min_{z\in\left\{0,1\right\}^n} \gamma\left\|z\right\|_0 + \left\|Bz\right\|_1 \\
    \mathrm{s.t.} \ y_i = \bigvee_{j=1}^{n}{M_{i,j}z_j},\ i=1,\ldots,m, \\
  \end{gathered}
\end{equation}
where $M$ is the pooling matrix and $\gamma>0$. The first term promotes the sparsity of $x$, whereas the second term encourages piecewise-constant solutions over the graph. However, the constraints in~(\ref{eq:recovery_problem_noiseless}) are not linear, we replace them with linear formulations~\cite{malioutov2012boolean}. Define the sets of the positive and negative pools as $\mathcal{P}=\left\{i\mid y_i=1\right\}$ and $\mathcal{N}=\left\{i\mid y_i=0\right\}$, respectively. The problem in~(\ref{eq:recovery_problem_noiseless}) can be equivalently formulated as
\begin{equation}\label{eq:recovery_problem_noiseless_linear}
  \begin{gathered}
    \min_{z\in\left\{0,1\right\}^n} \gamma\left\|z\right\|_0 + \left\|Bz\right\|_1 \\
    \begin{aligned}
      \mathrm{s.t.} \ & M_{i,:}z=0,\ i\in\mathcal{N}, \\
      & M_{i,:}z\geq 1,\ i\in\mathcal{P},
    \end{aligned}
  \end{gathered}
\end{equation}
where $M_{i,:}$ represents the $i$-th row of $M$. It is obvious that any sample appearing in negative pools can be directly declared negative. The remaining samples form the candidate set, which is given by
\begin{equation}\label{eq:candidate_set}
  \mathcal{U}=\left\{j\in\mathcal{V}\mid M_{i,j}=0\ \text{for}\ \text{all}\ i\in\mathcal{N}\right\}.
\end{equation}

In the noisy setting, each test outcome is independently flipped with probability $\epsilon\in\left(0,1/2\right)$, and the observed test outcomes $\tilde{y}$ may differ from the true results. Accordingly, we declare a sample negative only if it participates in many observed negative pools. Let $\tilde{\mathcal{P}}=\left\{i\mid \tilde{y}_i=1\right\}$ and $\tilde{\mathcal{N}}=\left\{i\mid \tilde{y}_i=0\right\}$. For each sample $j\in\left\{1,\ldots,n\right\}$, define $\tilde{N}_j = \sum_{i\in \tilde{\mathcal{N}}} {M_{i,j}}$, which counts the number of observed negative pools containing sample $j$. Given a threshold $\tau$, sample $j$ with $\tilde{N}_j \geq \tau$ is declared negative and excluded from the candidate set. The resulting candidate set is given by
\begin{equation}
  \mathcal{U}_{\epsilon}=\left\{j\in\mathcal{V}\mid \tilde{N}_j <\tau \right\}.
\end{equation}

In addition, we use the Hamming distance to account for measurement inconsistencies. For any $z\in\left\{0,1\right\}^n$, let $\hat{y}\left(z\right)$ denote the corresponding noiseless measurement vector, which is given by
\begin{equation}
  \hat{y}_i\left(z\right) = \bigvee_{j=1}^{n}{M_{i,j}z_j},\ i=1,\ldots,m.
\end{equation}
The normalized Hamming distance between $\tilde{y}$ and  $\hat{y}\left(z\right)$ is defined as
\begin{equation}
  \begin{aligned}
    d_H\left(\tilde{y},\hat{y}\left(z\right)\right) = \frac{1}{m}\sum_{i=1}^{m} {\mathbb{I}\left(\tilde{y}_i \neq \hat{y}_i\left(z\right)\right)} 
    = \frac{1}{m}\left(\sum_{i\in\tilde{\mathcal{N}}}{\mathbb{I}\left(M_{i,:}z \geq 1\right)} + \sum_{i\in\tilde{\mathcal{P}}}{\mathbb{I}\left(M_{i,:}z = 0\right)}\right) 
  \end{aligned}
\end{equation}
which measures the fraction of inconsistencies between the observed and predicted test outcomes. Thus, the sparse recovery problem is formulated as follows:
\begin{equation}\label{eq:recovery_problem_noisy}
  \begin{gathered}
    \min_{z\in\left\{0,1\right\}^n} \gamma\left\|z\right\|_0 + \left\|Bz\right\|_1 + \mu d_H\left(\tilde{y},\hat{y}\left(z\right)\right) \\
    \mathrm{s.t.} \ z_j=0,\ j\notin \mathcal{U}_{\epsilon},\\
  \end{gathered}
\end{equation}
where $\mu>0$ controls the penalty on measurement inconsistency. The proposed decoder accounts for measurement inconsistency through the Hamming-distance term, whereas the sparsity and graph total variation terms promote the sparsity and the piecewise-constant structure of $x$, respectively.

\subsection{Theoretical Analysis}
\subsubsection{Sufficient Condition for Exact Recovery in the Noiseless Setting}
We first present the sufficient condition on the number of tests required for exact recovery under the Bernoulli pooling design in the noiseless setting.

\begin{theorem}[Sufficient Condition for Noiseless Recovery]\label{thm:sufficient condition for noiseless recovery}
  Consider the system model in~(\ref{eq:system_noiseless}), where $x\in\left\{0,1\right\}^{n}$ is a $K$-sparse signal with support $\mathcal{I}\in \mathcal{F}_{K,k}$, and $M\in\left\{0,1\right\}^{m\times n}$ is a Bernoulli random matrix with $p=1/\left(1+K\right)$. Define $\lambda = 1+d_{max}/\gamma$ and let $\beta\geq \lambda$. Then with probability at least $1-\delta$, $x$ is the unique minimizer of~(\ref{eq:recovery_problem_noiseless_linear}), provided that
  \begin{equation}
    m\geq \max\left\{m_1,m_2,m_3,m_4\right\},
  \end{equation}
  where 
  \begin{equation}
    \begin{aligned}
      & m_1=e(K+1)\log\left({4\left(n-K\right)}/{\delta \beta K}\right), \\
      & m_2=e(K+1)\log\left({4Kd_{max}}/{\delta}\right), \\
      & m_3=8eK\log\left({4K^2}/{\delta}\right),\\
      & m_4=2e^{\beta+1}K\left(\log\left({4K^2}/{\delta}\right) + \lambda \log\left(\beta K+1\right)\right).
    \end{aligned}
  \end{equation}
\end{theorem}

The proof proceeds in three steps. We first establish that, the following statements hold with high probability: (i) the candidate set $\mathcal{U}$ contains at most $\left(1+\beta\right)K$ samples, (ii) all negative neighbors of infections are eliminated. We then prove that, with high probability, any $\mathcal{S}\subseteq\mathcal{U}$ with $\mathcal{S}\neq\mathcal{I}$ either violates the constraints or attains a larger objective function value than $\mathcal{I}$. The detailed proof is given in Appendix~\ref{proof:sufficient condition for noiseless recovery}.

\begin{remark}
  The above results demonstrate the benefit of incorporating the graph-based prior into group testing. For fixed $\beta$, $\lambda$ and $\delta$, if $n \gg K^2 d_{max}$, the number of tests required for exact recovery is $\mathcal{O}\left(K\log\left({n}/{K}\right)\right)$. Notably, this scaling is of the same order as the classical information-theoretic lower bound under the combinatorial prior. 
\end{remark}

\subsubsection{Performance Analysis of Graph-Aware Pooling Design}  
In the graph-aware pooling design, the pooling matrix $M$ is generated from $P^*$, i.e., $M_{i,j}\sim Bernoulli\left(P_{i,j}^{\ast}\right)$. The entries of $P^{\ast}$ are given by $P^{\ast}_{i,j} = u_i v_j\exp\left(-C_{i,j}/\varepsilon\right)$, where $C_{i,j}$ is the shortest-path distance between the representative vertex $c_i\in\mathcal{C}$ and the vertex $j$. The variables $u$ and $v$ are determined by the marginal constraints in~(\ref{eq:variables_constraints}). Let $v_{max}=\max_{j\in \{1,\ldots,n\}} {v_{j}}$ and $v_{min}=\min_{j\in \{1,\ldots,n\}} {v_{j}}$.
Define
\begin{equation}
  \theta=\varepsilon\log\left(\frac{v_{max}}{v_{\min}}\right),
\end{equation}
and
\begin{equation}
    \Gamma = \min \left\{n, \frac{1-\left(d_{max}\exp\left(-1/\varepsilon\right)\right)^{D+1}}{1-d_{max}\exp\left(-1/\varepsilon\right)}\right\}, 
\end{equation}
for $d_{max}\neq \exp\left({1}/{\varepsilon}\right)$, where $D$ is the diameter of $\mathcal{G}$. If $d_{max}=\exp\left(1/\varepsilon\right)$, $\Gamma=D+1$. For each sample $j\notin\mathcal{I}$, let $q_j^{OT}$ and $q_j^{B}$ denote the probabilities that sample $j$ remains in the candidate set $\mathcal{U}$ under the graph-aware and Bernoulli pooling designs, respectively. We prove that, the graph-aware design increases the probability of eliminating negative samples, thereby yielding a smaller candidate set $\mathcal{U}$.

\begin{theorem}\label{thm:advantage_OT}
  Suppose that, for every sample $j\notin\mathcal{I}$, there exists at least $r$ representative vertices $c_i\in\mathcal{C}$ satisfying 
  \begin{equation}\label{eq:C_constraints}
    \begin{gathered}
      C_{i,j}\leq R_1, \\
      C_{i,j'}-C_{i,j}\geq R_2, \ for\ all\ j'\in\mathcal{I}.
    \end{gathered}
  \end{equation}
  If 
  \begin{equation}
    \begin{gathered}
      \varepsilon\log\left(\frac{eKn}{\left(K+1\right)\Gamma}\right) - \theta > R_1 > \varepsilon\log\left(\frac{n}{\left(K+1\right)\Gamma}\right) - \theta, \\
      R_2 > \epsilon\log\left(\frac{K}{1-\frac{\left(K+1\right)\Gamma}{eKn}\exp\left(\frac{R_1+\theta}{\varepsilon}\right)}\right) + \theta, \\
      r > \frac{m(K+1)\Gamma\exp\left(\frac{R_1 +\theta}{\varepsilon}\right)}{eKn\left(1-K\exp\left(-\frac{R_2-\theta}{\varepsilon}\right)\right)},
    \end{gathered}
  \end{equation}
  then $q_{j}^{OT}< q_j^{B}$ holds for any sample $j\notin\mathcal{I}$. 
\end{theorem}

\begin{proof}
  Suppose that, for every sample $j\notin\mathcal{I}$, there exists at least $r$ representative vertices $c_i\in\mathcal{C}$ satisfying~(\ref{eq:C_constraints}). Given a sample $j\notin\mathcal{I}$, the probability that it remains in the candidate set under the Bernoulli pooling design is
  \begin{equation}
    q_j^{B}=\left(1-p(1-p)^K\right)^m \geq \left(1-\frac{1}{eK}\right)^m.
  \end{equation}  
  Under the graph-aware pooling design, the probability $q_j^{OT}$ satisfies
  \begin{equation}\label{eq:p_OT_inequality}
    \begin{aligned}
      q_j^{OT} & = \prod_{i=1}^{m}\left(1-P_{i,j}^{\ast}\prod_{j'\in \mathcal{I}}\left(1-P_{i,j'}^{\ast}\right)\right) \\
      & \leq \left(1-\frac{1}{m}\sum_{i=1}^{m}{P_{i,j}^{\ast}\prod_{j'\in \mathcal{I}}\left(1-P_{i,j'}^{\ast}\right)}\right)^m \\
      & \leq \left(1-\frac{1}{m}\sum_{i=1}^{m}{P_{i,j}^{\ast}\left(1-\sum_{j'\in\mathcal{I}}{P_{i,j'}^{\ast}}\right)}\right)^m.
    \end{aligned}
  \end{equation}

  According to~(\ref{eq:P_form}) and~(\ref{eq:variables_constraints}), we have
  \begin{equation}
      P_{i,j}^{\ast} = \frac{v_j\exp\left(-C_{i,j}/\varepsilon\right)np}{\sum_{j'=1}^{n}{v_{j'}\exp\left(-C_{i,j'}/\varepsilon\right)}}.
  \end{equation}
  For a graph with maximum degree $d_{max}$ and diameter $D$, the number of vertices at distance $h$ from any vertex is at most $d_{max}\left(d_{max}-1\right)^{h-1}$. Thus,
  \begin{equation}
    \begin{aligned}
      \sum_{j'=1}^{n}{\exp\left(-\frac{C_{i,j'}}{\varepsilon}\right)} \leq 1 + \sum_{h=1}^{D} \left(d_{max}\exp\left(-\frac{1}{\varepsilon}\right)\right)^h.
    \end{aligned}
  \end{equation}
  By the definition of $\Gamma$, we have 
  \begin{equation}
    \sum_{j'=1}^{n}{v_{j'}\exp\left(-\frac{C_{i,j'}}{\varepsilon}\right)} \leq v_{max}\sum_{j'=1}^{n}{\exp\left(-\frac{C_{i,j'}}{\varepsilon}\right)} \leq v_{max}\Gamma,
  \end{equation}
  and 
  \begin{equation}
    P_{i,j}^{\ast} \geq \frac{npv_{min}}{\Gamma v_{max}}\exp\left(-\frac{C_{i,j}}{\varepsilon}\right) = \frac{np}{\Gamma}\exp\left(-\frac{C_{i,j}+\theta}{\varepsilon}\right).
  \end{equation}

  For a vertex $c_i\in\mathcal{C}$ satisfying~(\ref{eq:C_constraints}), the probability $P_{i,j}^{\ast}$ is lower bounded by
  \begin{equation}\label{eq:P_inequality}
    P_{i,j}^{\ast} \geq \frac{np}{\Gamma}\exp\left(-\frac{R_1 + \theta}{\varepsilon}\right) .
  \end{equation}
  Then, for every sample $j'\in\mathcal{I}$, we have
   \begin{equation}
    \begin{aligned}
      \frac{P_{i,j'}^{\ast}}{P_{i,j}^{\ast}}  = \frac{v_{j'}}{v_j}\exp\left(-\frac{C_{i,j'}-C_{i,j}}{\varepsilon}\right) 
      \leq \frac{v_{max}}{v_{min}}\exp\left(-\frac{R_2}{\varepsilon}\right) = \exp\left(-\frac{R_2-\theta}{\varepsilon}\right).
    \end{aligned}
   \end{equation}
   Thus, 
   \begin{equation}\label{eq:P_sum_inequality}
    \sum_{j'\in\mathcal{I}}{P_{i,j'}^{\ast}} \leq K \exp\left(-\frac{R_2-\theta}{\varepsilon}\right).
   \end{equation}

   Combining~(\ref{eq:p_OT_inequality}),~(\ref{eq:P_inequality}) and~(\ref{eq:P_sum_inequality}), the probability $q_{j}^{OT}$ satisfies
   \begin{equation}
    \begin{aligned}
     q_j^{OT} \leq  \left(1-{\frac{rnp}{m\Gamma}\exp\left(-\frac{R_1 + \theta}{\varepsilon}\right)\left(1-K \exp\left(-\frac{R_2-\theta}{\varepsilon}\right)\right)}\right)^m,
    \end{aligned} 
   \end{equation}
   which is lower than $q_{j}^{B}$ for all $j\notin \mathcal{I}$ by the stated conditions on $R_1$,$R_2$ and $r$.
\end{proof}

The conditions in Theorem~\ref{thm:advantage_OT} require that, for each negative sample, there are sufficiently many representative vertices that are close to it but relatively far from the positive samples. Specifically, a smaller $R_1$ imposes a tighter upper bound on the distance between the negative sample and the representative vertex, making the negative sample more likely to be included in the corresponding pool. Meanwhile, a larger $R_2$ requires the positive samples to be farther from the same representative vertex than the negative sample, making them less likely to be in the same pool. These conditions are mild because they need to hold only for a sufficient number of representative vertices for each negative sample. Consequently, the graph-aware pooling design increases the probability that a negative sample appears in a pool containing no positives, making it more likely to be eliminated from the candidate set $\mathcal{U}$ and reducing the feasible search space for decoding.

\subsubsection{Sufficient Condition for Exact Recovery in the Noisy Setting}
We present the sufficient condition on the number of tests required for exact recovery under the Bernoulli pooling design in the noisy setting.

\begin{theorem}[Sufficient Condition for Noisy Recovery]\label{thm:sufficient condition for noisy recovery}
  Consider the system model in~(\ref{eq:system_noisy}) with flip probability $\epsilon\in\left(0,1/2\right)$, where $x\in\left\{0,1\right\}^{n}$ is a $K$-sparse signal with support $\mathcal{I}\in \mathcal{F}_{K,k}$, and $M\in\left\{0,1\right\}^{m\times n}$ is a Bernoulli random matrix with $p=1/\left(1+K\right)$. Define $\lambda = 1+d_{max}/\gamma$. Let $\beta\geq \lambda$ and $\mu \geq 2e^{\beta + 1}K(d_{max}+\gamma)/\left(1-2\epsilon\right)$. Then with probability at least $1-\delta$, $x$ is the unique minimizer of~(\ref{eq:recovery_problem_noisy}), provided that
  \begin{equation}
    m\geq \max\left\{m_1^{\epsilon},m_2^{\epsilon},m_3^{\epsilon}\right\},
  \end{equation}
  where 
  \begin{equation}
    \begin{aligned}
      & m_1^{\epsilon}= \frac{8e\left(1+(e-2)\epsilon\right)(K+1)}{(1-2\epsilon)^2} \log\left({3(n-K)}/{\delta \beta K}\right),\\
      & m_2^{\epsilon}=\frac{4e\left(1+(6e-2)\epsilon\right)(K+1)}{3(1-2\epsilon)^2} \log\left({3K}/{\delta}\right),\\
      & m_3^{\epsilon}=\frac{16e^{\beta+1}(2-\epsilon)K}{3(1-2\epsilon)^2}\log\left({n}/{\log\left(1+\delta/3\right)}\right).\\
    \end{aligned}
  \end{equation}
\end{theorem}

The proof is analogous to that of the noiseless case. We first prove the following statements: (i) the candidate set $\mathcal{U}_{\epsilon}$ contains at most $\left(1+\beta\right)K$ samples, (ii) infections are not eliminated. We then establish that, with high probability, every $\mathcal{S}\subseteq\mathcal{U}_{\epsilon}$ with $\mathcal{S}\neq\mathcal{I}$ attains a larger objective function value than $\mathcal{I}$. The detailed proof is given in Appendix~\ref{proof:sufficient condition for noisy recovery}.

\begin{remark}
  For fixed $\beta$, $\lambda$ and $\delta$, the number of tests required scales as $K\log n$. Compared with the noiseless case, the gap in scaling is primarily due to relaxations used in the proof. Deriving a tighter bound remains an open problem.
\end{remark}

\subsection{Convex Relaxation and Optimization}
Although we have established theoretical guarantees for the problems in~(\ref{eq:recovery_problem_noiseless_linear}) and~(\ref{eq:recovery_problem_noisy}), they are NP-hard due to the $\ell_0$-norm and the binary constraint. To obtain tractable formulations, we consider the following convex relaxations. In the noiseless setting, the problem in~(\ref{eq:recovery_problem_noiseless_linear}) can be relaxed as follows~\cite{malioutov2012boolean}:
\begin{equation}\label{eq:recovery_problem_noiseless_convex}
  \begin{gathered}
    \min_{z\in\left[0,1\right]^n} \gamma\left\|z\right\|_1 + \left\|Bz\right\|_1 \\
    \mathrm{s.t.} \ M_{i,:}z=0,\ i\in\mathcal{N}; \quad M_{i,:}z\geq 1,\ i\in\mathcal{P}. \\
  \end{gathered}
\end{equation}

To solve~(\ref{eq:recovery_problem_noiseless_convex}), we first apply COMP~\cite{aldridge2014group} to reduce the search space. Any sample appearing in negative pools must be negative and is removed from the candidate set. Then, we solve the optimization problem over the candidate set $\mathcal{U}$ by the alternating direction method of multipliers (ADMM)~\cite{boyd2011distributed,parikh2014proximal}. Finally, the binary estimate is obtained as $\hat{x}=\mathbb{I}\left(\hat{z}>0.5\right)$, where $\hat{z}$ denotes the solution to~(\ref{eq:recovery_problem_noiseless_convex}).

In the noisy setting, the Hamming-distance term in~(\ref{eq:recovery_problem_noisy}) is non-convex due to the indicator function. For a binary vector $z$, $\mathbb{I}\left(M_{i,:}z = 0\right)$ is equal to $\left( 1-M_{i,:}z\right)_{+}$, whereas $\mathbb{I}\left(M_{i,:}z \geq 1\right)$ is upper bounded by $M_{i,:}z$. Thus, we obtain the relaxation of~(\ref{eq:recovery_problem_noisy}):
\begin{equation}\label{eq:recovery_problem_noisy_convex}
  \begin{gathered}
    \min_{z\in\left[0,1\right]^n} \gamma\left\|z\right\|_1 + \left\|Bz\right\|_1 + \frac{\mu}{m} \left(\sum_{i\in \tilde{\mathcal{N}}} {M_{i,:}z} + \sum_{i\in \tilde{\mathcal{P}}} {\left( 1-M_{i,:}z\right)_{+}}\right) \\
    \mathrm{s.t.} \ z_j=0,\ j\notin \mathcal{U}_{\epsilon}.\\
  \end{gathered}
\end{equation}

This problem is solved in the same way. The search space is restricted to $\mathcal{U}_{\epsilon}$, and then we apply ADMM. The binary estimate is obtained by thresholding the solution to~(\ref{eq:recovery_problem_noisy_convex}).

\section{Numerical Experiments}\label{sec:numerical experiments}
In this section, we conduct numerical simulations to evaluate the performance of our method. We first describe the experimental setup, including the datasets, parameter configurations, and evaluation metrics. Then, we present the results under different pooling designs and decoders.

\subsection{Experimental Setup}
\subsubsection{Datasets}
The experiments are conducted on synthetic and real-world networks. The synthetic networks include Erd\H{o}s-R\'enyi (ER) random network, small-world network, Barab\'as-Albert (BA) scale-free network~\cite{barabasi1999emergence}, and stochastic block model (SBM)~\cite{holland1983stochastic}, each consisting of $n=1000$ vertices. The ER random network is generated with an edge connection probability of ${6}/{(n-1)}$. The small-world network is constructed according to the Watts-Strogatz model~\cite{watts1998collective}, where each vertex is initially connected to its six nearest neighbors and then each edge is rewired with probability of $0.1$. For the BA network, the graph is generated by adding new vertices, each with two edges that are preferentially attached to existing high-degree nodes. The SBM network is generated with 10 communities with the intra-community and inter-community connection probabilities set to $0.1$ and $0.005$, respectively.

We also consider two empirical contact networks. The SGinfectious dataset~\cite{isella2011s} is provided by the ScoioPatterns sensing platform, with active Radio-Frequency Identification Devices embedded in badges to record face-to-face proximity relations of participants\cite{chai2021information}. We use 15 days of contact records in our experiments. The Sexual contact dataset~\cite{rocha2011simulated, ru2023inferring} is an empirical temporal network of sexual contacts in Brazil, and we use 90 days of records. The adjacency matrix is constructed from the time aggregated graph of the temporal network. The detailed statistical characteristics of networks are listed in Table~\ref{tab:network statistical characteristics}.
\begin{table}[htbp]
\caption{Statistics Characteristics of Networks} 
\label{tab:network statistical characteristics}
\centering
\begin{threeparttable}
\begin{tabular}{ccccc}
\hline
Network & Nodes & $d_{max}$ & $\left\langle d \right\rangle$ & Duration (days) \\
\hline
ER & 1000 & 6.15 & 16 & / \\
WS & 1000 & 6.00 & 9 & / \\
BA & 1000 & 3.99 & 69 & / \\
SBM & 1000 & 14.20  & 27 & / \\
SGinfectious & 3010 & 9.32 & 50 & 15 \\
Sexual & 3134 & 2.46 & 48 & 90 \\
\hline
\end{tabular}
\begin{tablenotes}
    \item $\left\langle d \right\rangle$ denotes the average degree of the network.
\end{tablenotes}
\end{threeparttable}
\end{table}

\subsubsection{Parameter Configurations}
We use the SI model to simulate infection propagation. The number of sources is drawn uniformly from $\left\{1,2,3\right\}$ and the transmission probability is drawn uniformly from $\left[0.1,0.5\right]$. The propagation terminates once the number of infected individuals reaches $nq$, where $q$ denotes the disease prevalence. The resulting number of infections is denoted by $K$ with $K\leq nq$. The infection states of the population are represented by a $K$-sparse vector $x\in\left\{0,1\right\}^n$, where $x_i=1$ if individual $i$ is infected and $x_i=0$ otherwise. Unless otherwise stated, we adopt the Bernoulli pooling design, where the entries of $M\in\left\{0,1\right\}^{m\times n}$ are independently drawn from $Bernoulli(p)$ with $p=1/(1+nq)$. The test outcomes are obtained according to~(\ref{eq:system_noiseless}), and in the noisy setting, according to~(\ref{eq:system_noisy}) with flip probability $\epsilon=0.02$. The optimization problem in~(\ref{eq:recovery_problem_noiseless_convex}) (or~(\ref{eq:recovery_problem_noisy_convex}) in the noisy setting) is solved by ADMM with $\gamma=0.01$ and $\mu=20$. For each value of $m$, we perform 200 independent trials, with both $x$ and $M$ regenerated in each trial. 

\subsubsection{Evaluation Metrics}
To evaluate the performance of group testing methods, we adopt the following metrics: success rate, recall (also known as the true positive rate), precision and F1 score. The success rate is defined as the fraction of trials in which all infected individuals are correctly identified. Recall, precision, and F1 score are used to provide a comprehensive evaluation of identification performance.

In experiments, we compare our decoders (GraphTV) with several non-adaptive group testing methods. Specifically, we compare with COMP, DD, SCOMP~\cite{aldridge2014group} in the noiseless setting, which are combinatorial group testing methods. COMP eliminates all samples appearing in negative pools and declares the remaining positive. DD further identifies a sample as positive if it is the sole member of a positive pool. SCOMP refines the results of DD by greedily adding candidate samples until all positive tests are explained. In the noisy setting, we compare with NCOMP~\cite{chan2011non}, NDD~\cite{scarlett2020noisy} and Noisy LP~\cite{malioutov2012boolean}. We also consider two graph-based decoders, C-LBP~\cite{nikolopoulos2023community} and Linearized QP~\cite{10002300}, where C-LBP is a loopy belief propagation based decoder that exploits the  community structure, and Linearized QP is a decoder based on Ising model. The parameters of these methods are set according to the configurations reported in the corresponding papers. For the pooling design, we compare the graph-aware pooling design with the Bernoulli and the constant column weight designs. In the constant column weight design, each individual is assigned to $mp$ pools selected uniformly at random.
\begin{figure*}[!t]
\centering
\subfloat[Recovery performance on WS network]{\includegraphics[width=\textwidth]{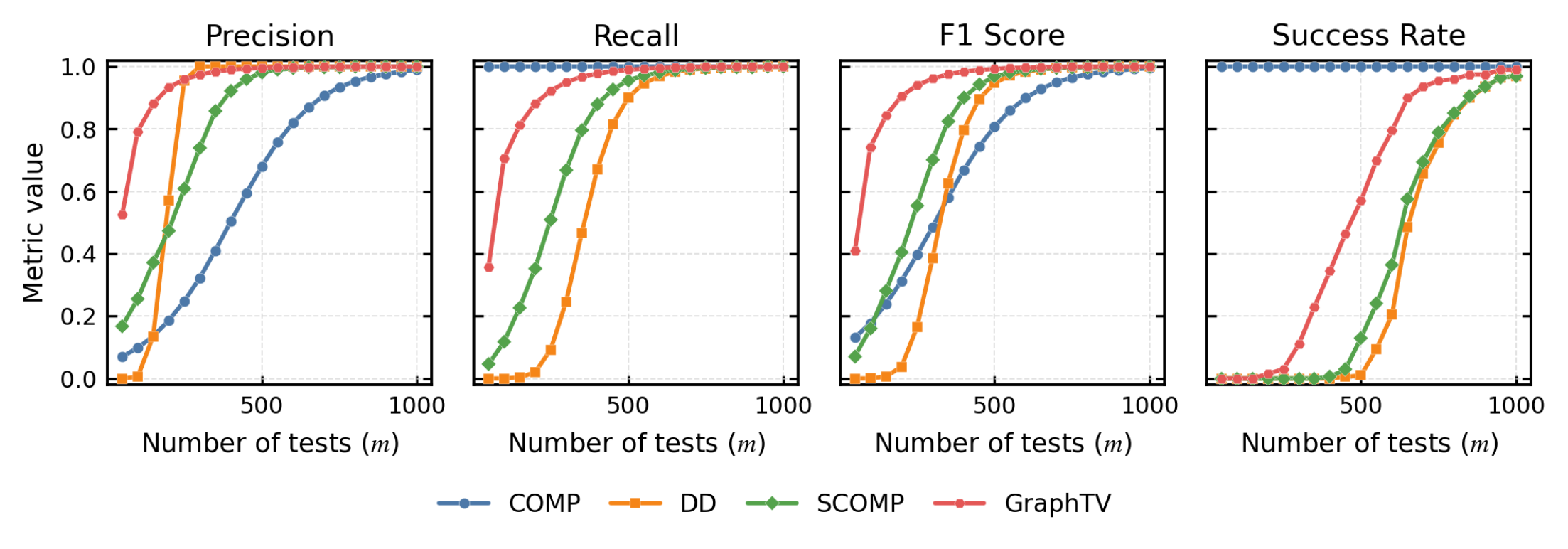}%
}
\quad
\subfloat[Recovery performance on SGinfectious network]{\includegraphics[width=\textwidth]{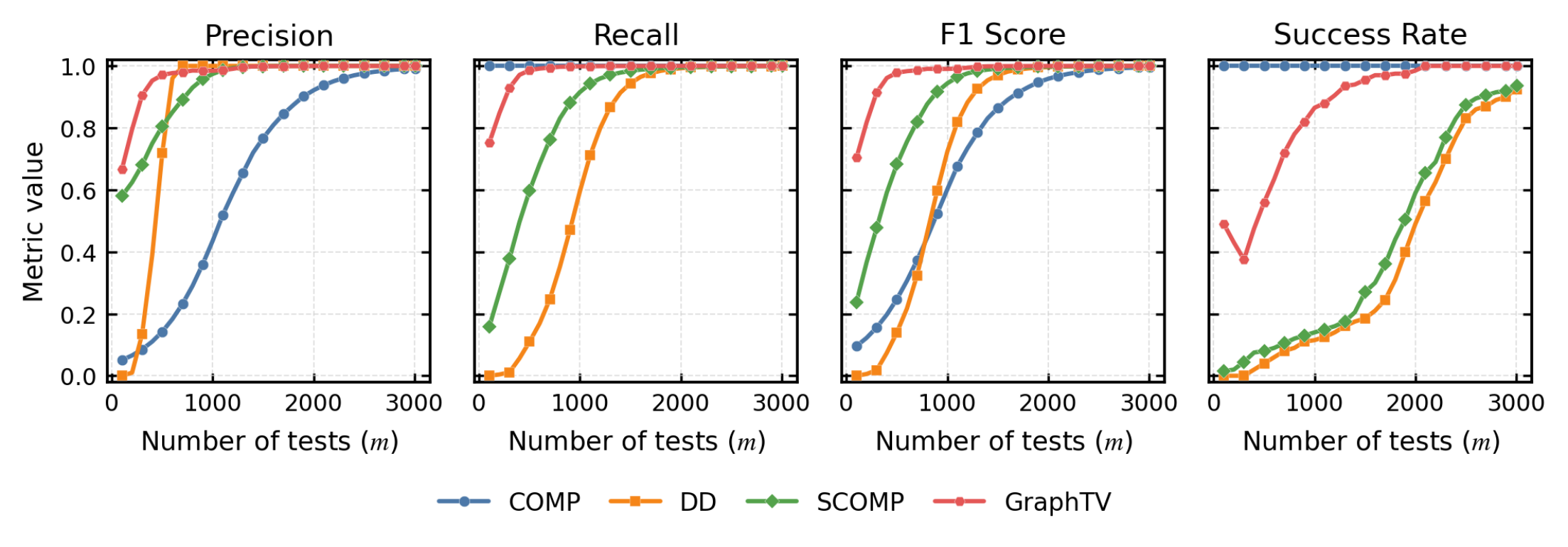}%
}
\caption{Performance comparison of different decoders under the Bernoulli pooling design in the noiseless setting with $q=0.05$.}
\label{fig:noiseless performance comparison}
\end{figure*}

\subsection{Experimental Results}
\subsubsection{Recovery Performance under Bernoulli Pooling Design}
We first evaluate the recovery performance of our decoder in the noiseless setting with $q=0.05$. Fig~\ref{fig:noiseless performance comparison} reports precision, recall, F1 score, and success rate versus the number of tests on the WS and SGinfectious networks. The proposed decoder consistently outperforms DD and SCOMP on both networks. On the WS network, when $m=500$, our decoder improves the success rate by $44\%$ over SCOMP, and also achieves higher recall, precision, and F1 score. Similar improvements over DD are also observed. On the SGinfectious network, when $m=2000$, our decoder improves the success rate, F1 score, and recall over SCOMP by $39.5\%$, $0.21\%$, and $0.44\%$, respectively. Although COMP achieves a success rate and recall of one across all values of $m$, its conservative strategy produces a large number of false positives, resulting in substantially lower precision and F1 score.

We also evaluate the recovery performance in the noisy setting on the ER, WS, SGinfectious, and Sexual contact networks with $q=0.02$. The results are shown in Fig~\ref{fig:noisy performance comparison}. Although the recall of our decoder is slightly lower than that of Noisy LP, it achieves the highest precision and F1 score on all networks. This indicates that our decoder yields fewer false positives and achieves a better balance between precision and recall. In contrast, Noisy LP achieves higher recall at the expense of lower precision, resulting in more false positives. Overall, these results validate the effectiveness of our decoder in both noiseless and noisy settings, and highlight the benefit of exploiting the graph-based prior in decoding.
\begin{figure}[!t]
\centering
\includegraphics[width=\linewidth]{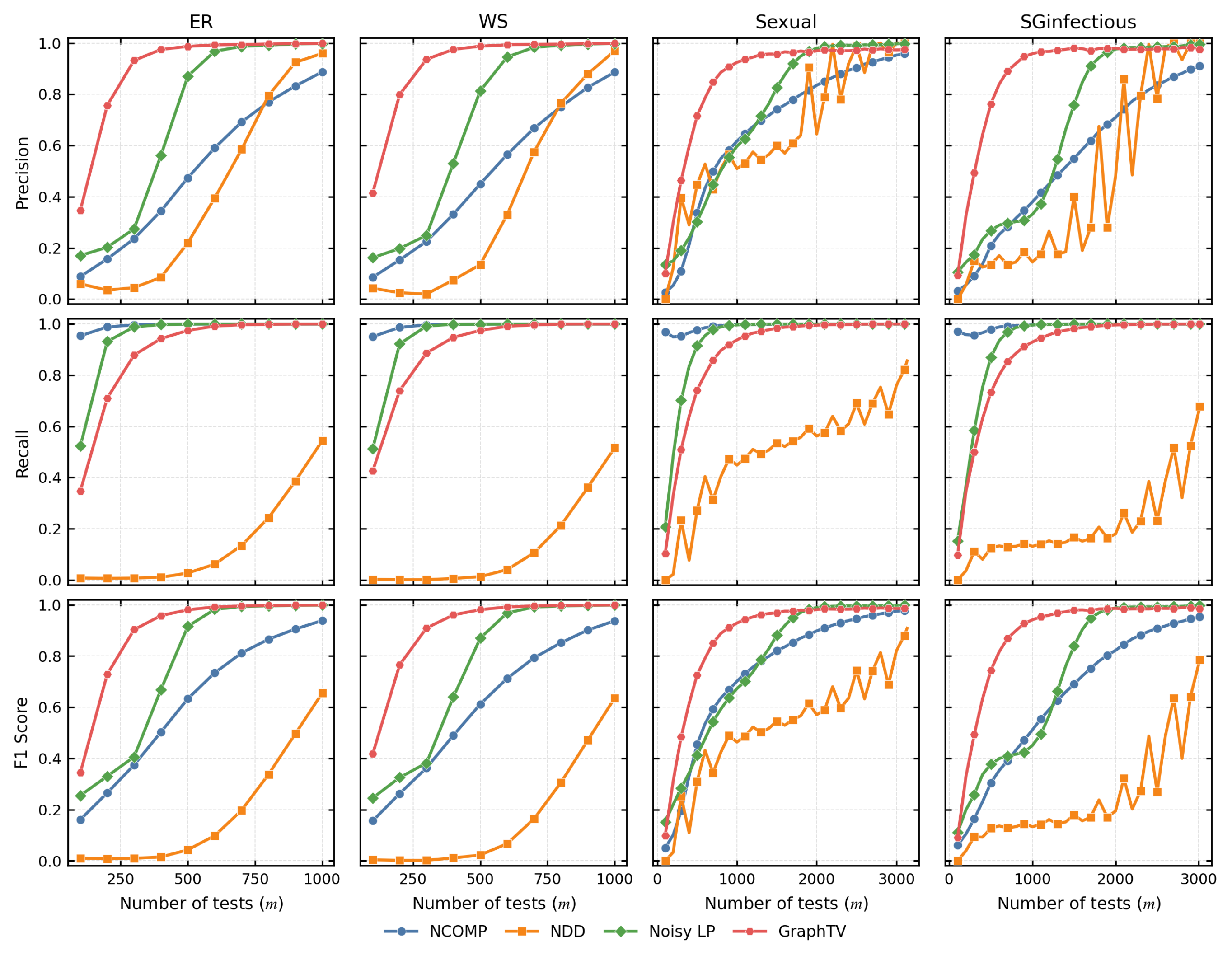}
\caption{Performance comparison of different decoders under the Bernoulli pooling design in the noisy setting with $q=0.02$.}
\label{fig:noisy performance comparison}
\end{figure}

\subsubsection{Comparison of Pooling Designs}
We next compare the proposed graph-aware pooling design with the Bernoulli and the constant column weight designs on synthetic networks. Fig.~\ref{fig:pooling design comparison}\subref{fig:pooling_compare_ER}-\subref{fig:pooling_compare_BA} reports the performance of SCOMP and the proposed decoder under three pooling designs, showing that the performance improvement is primarily due to the proposed pooling design rather than a specific decoding algorithm. The graph-aware design consistently achieves the best performance across all networks, where the improvement is particularly pronounced in success rate, followed by the constant column weight design, while the Bernoulli design performs worst. Under comparable pool sizes and numbers of pools per sample, the main difference among the pooling designs lies in the assignment of samples to pools. Random pooling designs assign samples independently of the graph structure, whereas our design exploits the graph structure and thereby favors assigning neighboring individuals to common pools. Consequently, infection clusters are concentrated within fewer pools, resulting in more negative pools and facilitating the elimination of negative samples.
\begin{figure}[!t]
\centering
\subfloat[Recovery performance on ER network]{\includegraphics[width=0.5\linewidth]{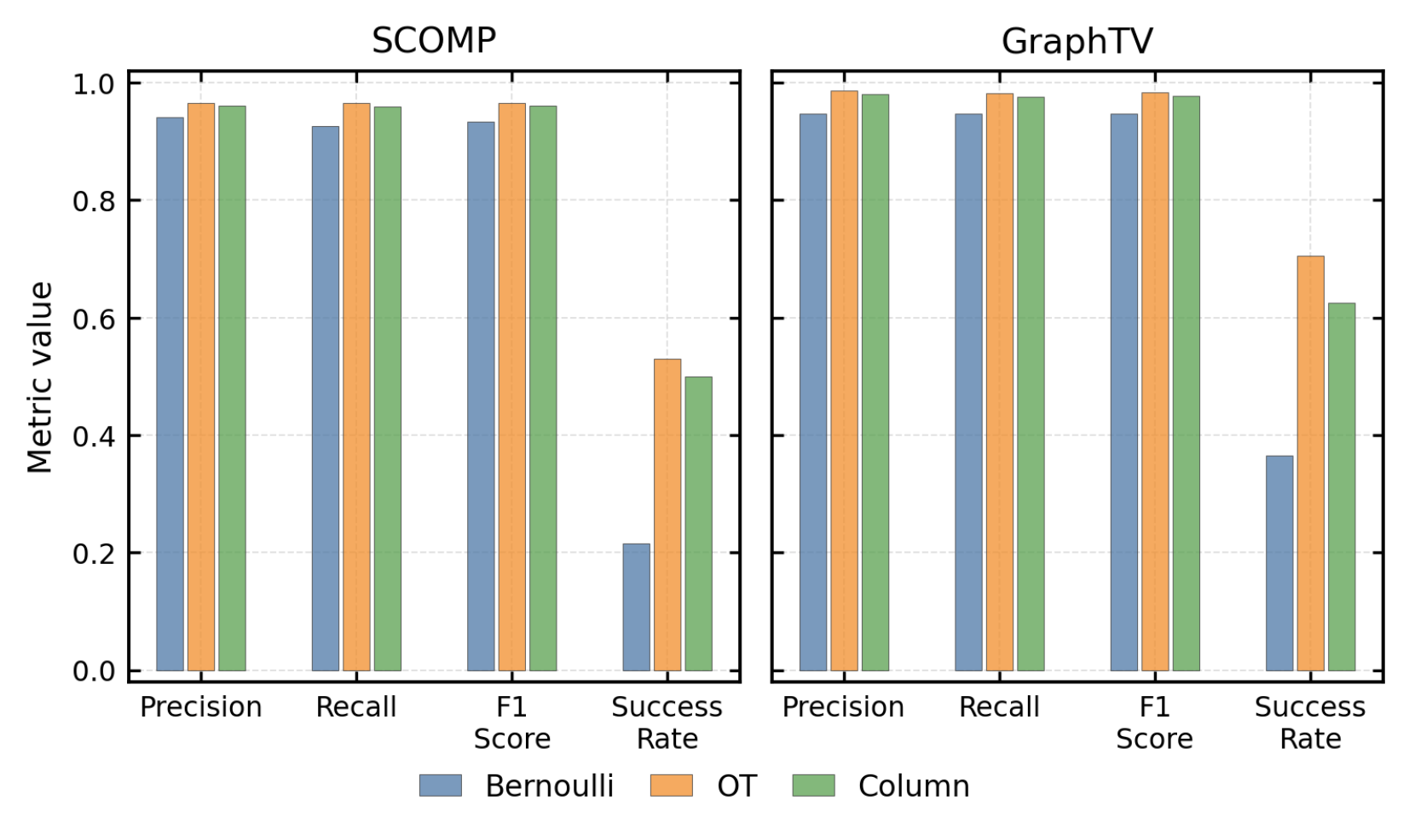}%
\label{fig:pooling_compare_ER}
}
\subfloat[Recovery performance on WS network]{\includegraphics[width=0.5\linewidth]{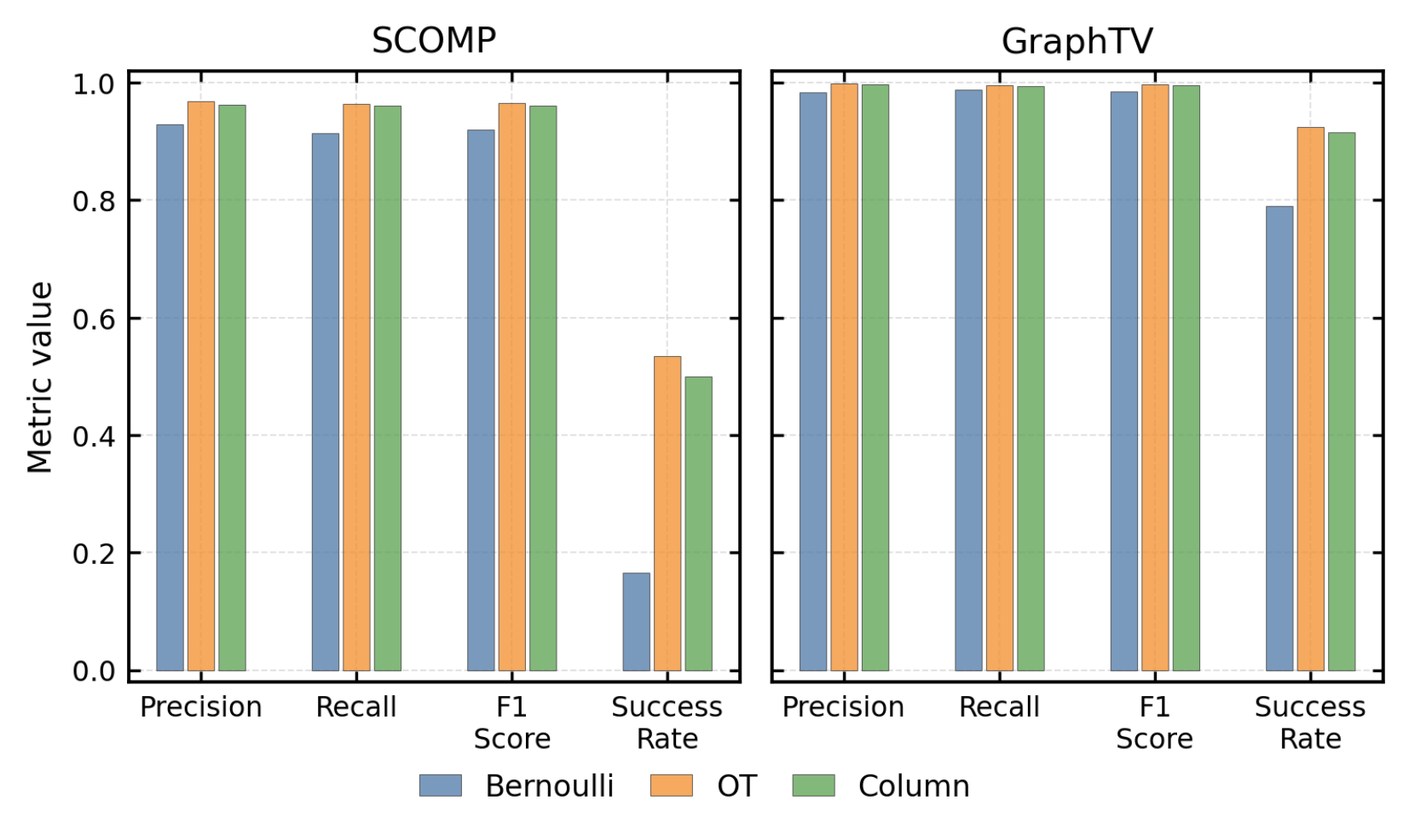}%
\label{fig:pooling_compare_WS}
}
\quad
\subfloat[Recovery performance on BA network]{\includegraphics[width=0.5\linewidth]{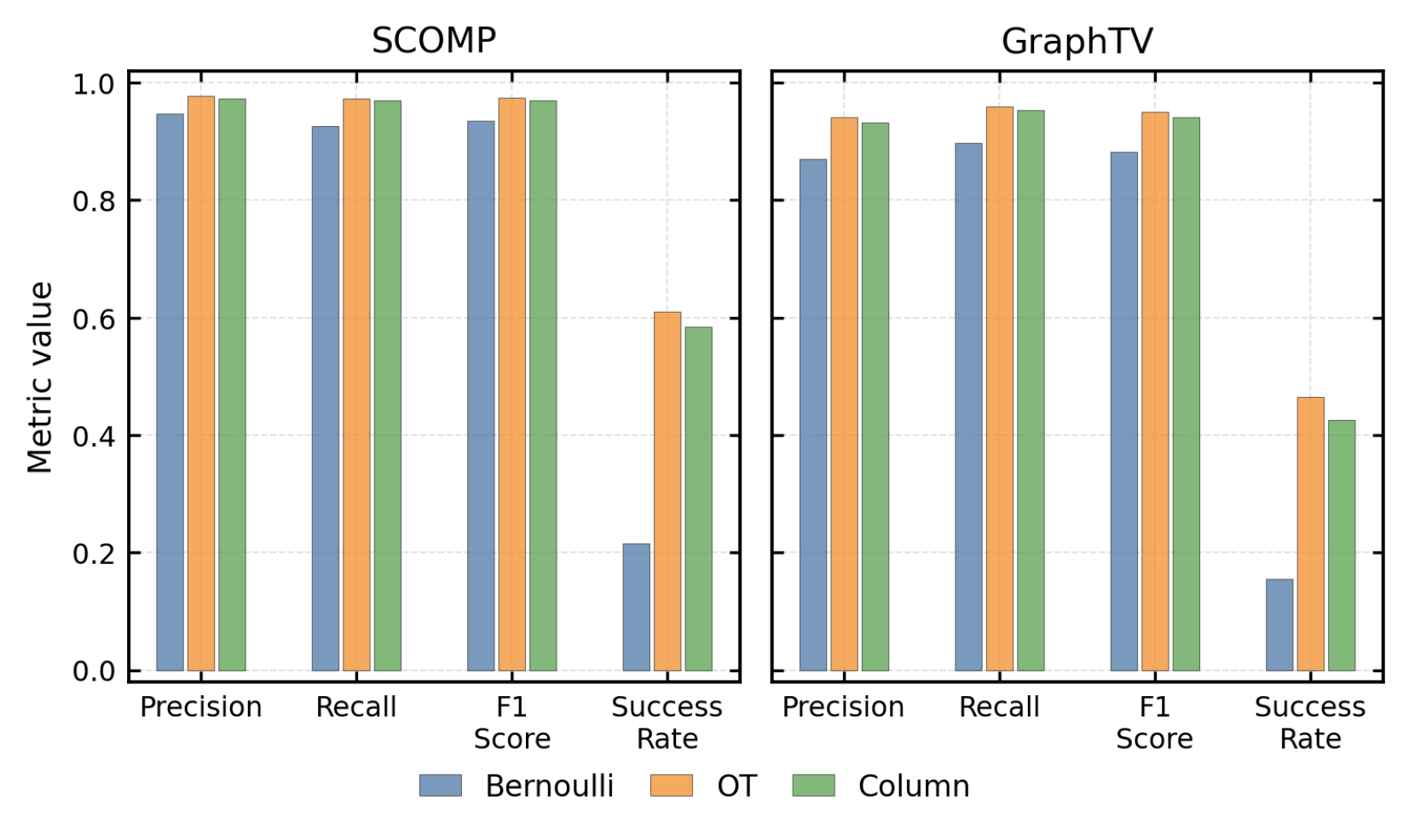}%
\label{fig:pooling_compare_BA}
}
\subfloat[Recovery performance on Sexual contact network]{\includegraphics[width=0.5\linewidth]{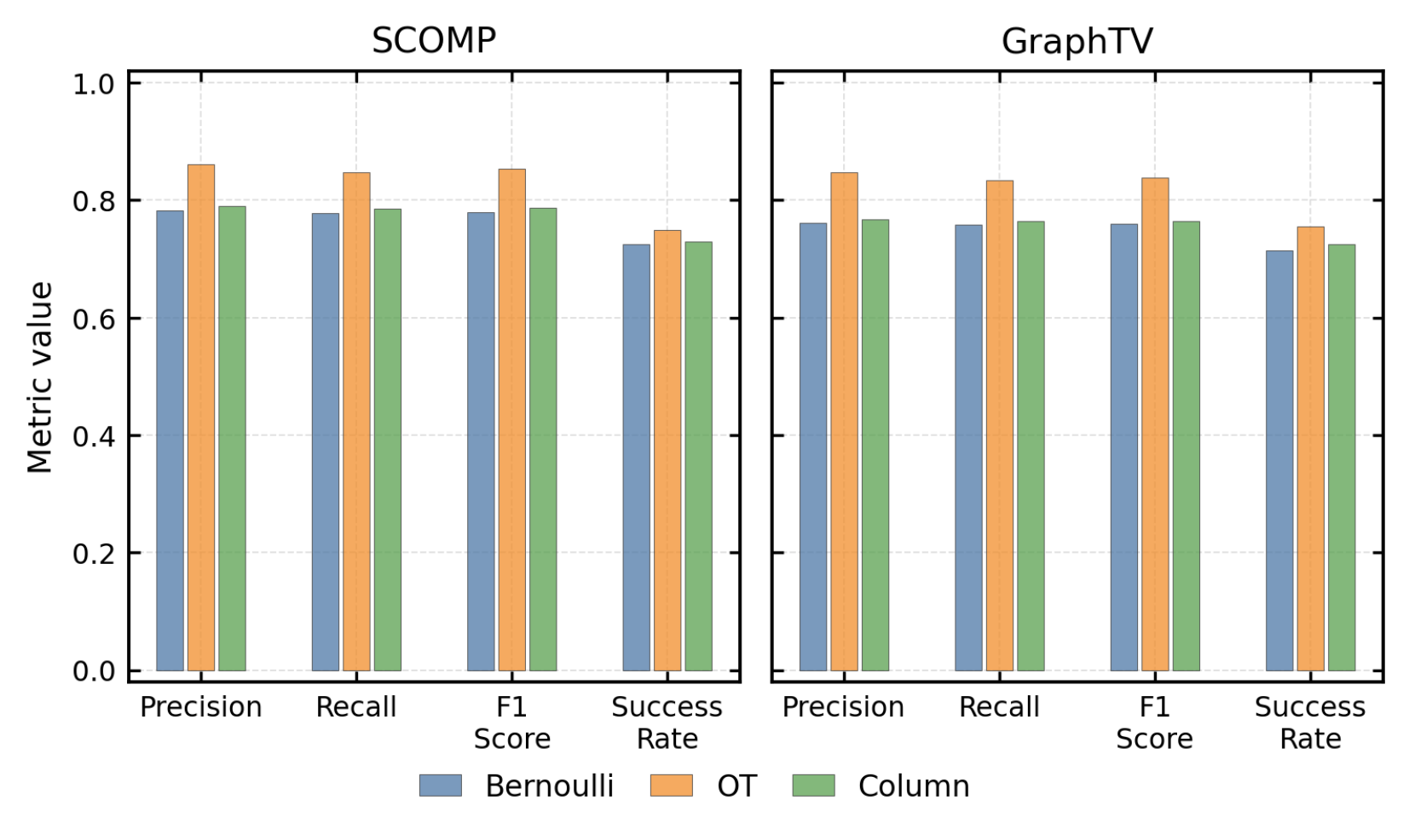}%
\label{fig:pooling_compare_Sexual}
}
\caption{Comparison of pooling designs in the noiseless setting. (a)-(c) Results on synthetic networks with $m=200$ and $q=0.02$. The parameter in the graph-aware pooling design is set to $\varepsilon=10$. (d) Results on the Sexual contact network with $m=500$. The parameter in the graph-aware pooling design is set to $\varepsilon=2$.}
\label{fig:pooling design comparison}
\end{figure}

We further evaluate the pooling designs on the Sexual contact network. The spreading process is initialized from three sources and propagates over the network for $90$ days with a transmission probability of $0.3$. The resulting disease prevalence $q$ is approximately between $0.003$ and $0.007$, and the pooling probability is set to $p=1/(1+nq)$. As shown in Fig.~\ref{fig:pooling design comparison}\subref{fig:pooling_compare_Sexual}, the graph-aware design also achieves better performance than the other pooling designs. This improvement stems from the graph-aware pooling strategy, which leverages graph distances to encourage individuals with correlated infection states to be pooled together, thereby providing more informative measurements.

\subsubsection{Comparison with Graph-based Methods}
Finally, we compare the proposed decoder under Bernoulli, constant column weight, and graph-aware pooling designs with Linearized QP and C-LBP, both implemented under Bernoulli design. Table~\ref{tab:performance of graph-based methods} summarizes the results on the ER, WS, BA, and SBM networks when $m=500$ and $q=0.05$. Under the same Bernoulli pooling design, the proposed decoder achieves comparable performance to Linearized QP and C-LBP. Furthermore, integrating the graph-aware pooling design improves the performance of our decoder. Specifically, on the WS network, although three pooling designs attain comparable precision and recall, GraphTV+OT achieves the highest success rate of $80.5\%$, demonstrating a substantial improvement over GraphTV+Bernoulli, Linearized QP and C-LBP by $19.5\%$, $21\%$, $52.5\%$, respectively. Similar improvements are observed on the other networks. Overall, our method (GraphTV+OT) achieves the best performance across all networks, with particularly significant improvements in success rate, demonstrating the benefit of exploiting the graph-induced correlations in both pooling design and decoding.
\begin{table}[!t]
\caption{Performance of Graph-based Methods\label{tab:performance of graph-based methods}}
\centering
\begin{tabular}{cccccc}
\hline
Network & Method & Precision & Recall & F1 Score & Success Rate \\
\hline
\multirow{5}{*}{\makecell[c]{ER \\ ($\varepsilon=2$)}} & GraphTV+Bernoulli & 0.9847 & 0.9703 & 0.9773 & 0.26 \\
& GraphTV+OT & \textbf{0.9954} & \textbf{0. 9875} & \textbf{0.9914} & \textbf{0.555} \\
& GraphTV+Column & 0.9950 & 0.9875 & 0.9912 & 0.525 \\
& Linearized QP & 0.9870 & 0.9705 & 0.9786 & 0.235 \\
& C-LBP & 0.9800 & 0.9669 & 0.9732 & 0.18 \\
\hline
\multirow{5}{*}{\makecell[c]{WS \\ ($\varepsilon=10$)}} & GraphTV+Bernoulli & 0.9970 & 0.9898 & 0.9933 & 0.61 \\
& GraphTV+OT & \textbf{0.9991} & \textbf{0.9961} & \textbf{0.9976} & \textbf{0.805} \\
& GraphTV+Column & 0.9985 & 0.9953 & 0.9969 & 0.8 \\
& Linearized QP & 0.9976 & 0.9898 & 0.9936 & 0.595 \\
& C-LBP & 0.9876 & 0.9735 & 0.9803 & 0.28 \\
\hline
\multirow{5}{*}{\makecell[c]{BA \\ ($\varepsilon=1$)}} & GraphTV+Bernoulli & 0.9691 & 0.9621 & 0.9655 & 0.14 \\
& GraphTV+OT & \textbf{0.9900} & \textbf{0.9879} & \textbf{0.9889} & \textbf{0.54} \\
& GraphTV+Column & 0.9870 & 0.9843 & 0.9856 & 0.455 \\
& Linearized QP & 0.9706 & 0.9622 & 0.9662 & 0.135 \\
& C-LBP & 0.9806 & 0.9664 & 0.9732 & 0.185 \\
\hline
\multirow{5}{*}{\makecell[c]{SBM \\ ($\varepsilon=1$)}} & GraphTV+Bernoulli & 0.9860 & 0.9668 & 0.9762 & 0.19 \\
& GraphTV+OT & \textbf{0.9967} & \textbf{0.9888} & \textbf{0.9927} & \textbf{0.55} \\
& GraphTV+Column & 0.9951 & 0.9856 & 0.9903 & 0.46 \\
& Linearized QP & 0.9877 & 0.9663 & 0.9768 & 0.185 \\
& C-LBP & 0.9782 & 0.9686 & 0.9732 & 0.18 \\
\hline
\end{tabular}
\end{table}

\section{Conclusion}\label{sec:conclusion}
In this paper, we studied graph-aware group testing that exploits correlations among individuals without requiring additional information beyond the graph structure. Specifically, we developed a graph-aware group testing framework that exploits localized infection clustering in pooling design, fundamental limits of recovery, and decoding. For pooling design, we proposed an optimal transport-based approach that incorporates graph proximity and pooling constraints. We proved that, under mild conditions, the proposed design can eliminate uninfected individuals with higher probability than the Bernoulli pooling design, thereby reducing the candidate set and the decoding complexity. We further investigated the fundamental limits of recovery by characterizing the family of possible infected sets induced by localized infection clustering, which results in a lower bound on the number of tests than that under the combinatorial prior. For decoding, we modeled the infection states of the population as a piecewise-constant graph signal and incorporated a graph total variation regularizer into the sparse recovery formulation. We established sufficient conditions for the proposed decoder under the Bernoulli design. Under mild conditions, the number of tests required for exact recovery is $\mathcal{O}\left(K\log\left(n/K\right)\right)$ in the noiseless case and $\mathcal{O}\left(K\log n\right)$ in the noisy case. Closing the gap in scaling is left for future work. Extensive simulations on synthetic and real-world networks demonstrated the effectiveness of the proposed graph-aware framework, which achieves better recovery performance than the existing methods. Future work will consider more challenging settings, including uncertain and time-varying contact graphs.

{\appendices
\section{Proof of  Theorem~\ref{thm:sufficient condition for noiseless recovery}}\label{proof:sufficient condition for noiseless recovery}
To prove Theorem~\ref{thm:sufficient condition for noiseless recovery}, we need the Chernoff's inequality.
\begin{lemma}[Chernoff's inequality~\cite{chung2006concentration}]\label{lemma:chernoff}
  Let $X_1,\ldots,X_n$ be independent random variables with $\mathbb{P}\left(X_i=1\right)=p_i$ and $\mathbb{P}\left(X_i=0\right)=1-p_i$. Let $X=\sum_{i=1}^{n}{X_i}$. Then, we have
  \begin{equation}
    \begin{aligned}
      & \mathbb{P}\left(X\leq \mathbb{E}\left[X\right]-t\right) \leq \exp\left(-\frac{t^2}{2\mathbb{E}\left[X\right]}\right), \\
      & \mathbb{P}\left(X\geq \mathbb{E}\left[X\right]+t\right) \leq \exp\left(-\frac{t^2}{2\left(\mathbb{E}\left[X\right] + t/3\right)}\right).
    \end{aligned}
  \end{equation}
\end{lemma}

To derive sufficient conditions for exact recovery, we reformulate the problem
in~(\ref{eq:recovery_problem_noiseless_linear}) with respect to the support of the binary vector $z$. The candidate set $\mathcal{U}$ is given in~(\ref{eq:candidate_set}). Let $\mathcal{R}=\mathcal{U}\backslash \mathcal{I}$ denote the set of negative samples in $\mathcal{U}$. For any $\mathcal{S}\subseteq\mathcal{V}$, define its edge boundary as
\begin{equation}
  \partial_{\mathcal{E}} \mathcal{S} = \{(u,v)\in \mathcal{E} \mid u\in \mathcal{S}, v\in \mathcal{V} \backslash \mathcal{S}\}.
\end{equation}
By the definition of graph total variation in~(\ref{eq:graph_total_variation}), we have $\left\|Bz\right\|_1 = \left|\partial_{\mathcal{E}} \mathcal{S}\right|$, where $\mathcal{S}=\operatorname{supp}\left(z\right)$. Thus, the problem in~(\ref{eq:recovery_problem_noiseless_linear}) is reformulated as 
\begin{equation}\label{eq:recovery_problem_noiseless_set}
  \begin{gathered}
    \min_{\mathcal{S}\subseteq \mathcal{U}} \gamma\left|\mathcal{S}\right| + \left|\partial_{\mathcal{E}} \mathcal{S}\right| \\
    \mathrm{s.t.} \ \sum_{j\in\mathcal{S}}{M_{i,j}\geq 1},\ i\in\mathcal{P}, \\
  \end{gathered}
\end{equation}
where $\mathcal{P}$ is the set of positive pools.

To prove Theorem~\ref{thm:sufficient condition for noiseless recovery}, we establish the following lemmas.
\begin{lemma}\label{lemma:candidate set properties}
  Consider the system model in~(\ref{eq:system_noiseless}), where $x$ is a $K$-sparse signal with support $\mathcal{I}\in\mathcal{F}_{K,k}$, and $M$ is a Bernoulli random matrix with $p=1/(K+1)$. Define $\lambda = 1+d_{max}/\gamma$ and let $\beta \geq \lambda$. Each of the following statements holds with probability at least $1-\delta/4$:
  \begin{itemize}
    \item [(i)] $\left|\mathcal{U}\right|\leq \left(1+\beta\right)K$, provided that
    \begin{equation}
      m \geq e(K+1)\log\left({4\left(n-K\right)}/{\delta \beta K}\right).
    \end{equation}
    \item [(ii)] $\mathcal{E}\left(\mathcal{I}, \mathcal{R}\right)=\varnothing$, provided that
    \begin{equation}
      m\geq e(K+1)\log\left({4Kd_{max}}/{\delta}\right),
    \end{equation}
    where $\mathcal{E}\left(\mathcal{I}, \mathcal{R}\right)= \{(u,v)\in \mathcal{E} \mid u\in \mathcal{I}, v\in \mathcal{R}\}$.
  \end{itemize}
\end{lemma}

\begin{proof}[Proof of Lemma~\ref{lemma:candidate set properties}]
  For a sample $j\notin \mathcal{I}$, the probability that it is in the candidate set is given by
  \begin{equation}
    \mathbb{P}\left(j\in\mathcal{U}\right)=\left(1-p\left(1-p\right)^K\right)^m.
  \end{equation}
  Thus, the number of negative samples in $\mathcal{U}$ satisfies
  \begin{equation}
    \mathbb{E}\left[\left|\mathcal{R}\right|\right] = (n-K)\left(1-p(1-p)^K\right)^m.
  \end{equation} 
  Applying Markov's inequality, yields
  \begin{equation}
    \begin{aligned}
      \mathbb{P}(\left|\mathcal{R}\right|>\beta K) &\leq \frac{(n-K)\left(1-p(1-p)^K\right)^m}{\beta K} \\
      & \leq \frac{(n-K)\exp\left(-m/\left(e(K+1)\right)\right)}{\beta K},
    \end{aligned}
  \end{equation}
  which is upper bounded by $\delta/4$ provided that
  \begin{equation}
    m \geq e(K+1)\log\left({4\left(n-K\right)}/{\delta \beta K}\right).
  \end{equation}
  This proves the part (i).
  
  For the part (ii), there exists edges between $\mathcal{I}$ and $\mathcal{R}$ only if at least one sample in $\mathcal{V}\backslash\mathcal{I}$ adjacent to $\mathcal{I}$ is not eliminated. The probability that this event occurs satisfies
  \begin{equation}
    \mathbb{P}\left(\mathcal{E}\left(\mathcal{I}, \mathcal{R}\right)\neq \varnothing\right) \leq Kd_{max}\left(1-p(1-p)^K\right)^m,
  \end{equation}
  which is upper bounded by $\delta/4$ provided that
  \begin{equation}
    m\geq e(K+1)\log\left({4Kd_{max}}/{\delta}\right).
  \end{equation}
  This completes the proof of Lemma~\ref{lemma:candidate set properties}.
\end{proof}

\begin{lemma}\label{lemma:test coverage}
  Under the same setting as Lemma~\ref{lemma:candidate set properties}, for each non-empty subset $\mathcal{D}\subseteq \mathcal{I}$ with $\left|\mathcal{D}\right|=d$, let $T_{\mathcal{D}}$ denote the number of pools that contain at least one sample from $\mathcal{D}$ and no sample from $\mathcal{I}\backslash\mathcal{D}$. For $d=1,\ldots,K$, define $p_d=\left(1-(1-p)^d\right)(1-p)^{K-d}$. Then, with probability with $1-\delta/4$, $T_{\mathcal{D}}\geq mp_d /2$ holds for every non-empty $\mathcal{D}\subseteq \mathcal{I}$, provided that
  \begin{equation}
    m\geq 8eK\log\left({4K^2}/{\delta}\right).
  \end{equation}
\end{lemma}

\begin{proof}[Proof of Lemma~\ref{lemma:test coverage}]
  Given a non-empty subset $\mathcal{D}\subseteq \mathcal{I}$ with $\left|\mathcal{D}\right|=d$. The probability that a pool contain at least one sample from $\mathcal{D}$ and no sample from $\mathcal{I}\backslash\mathcal{D}$ is 
  \begin{equation}
      p_d = \left(1-(1-p)^d\right)(1-p)^{K-d} \geq dp\left(1-p\right)^{K-1} \geq \frac{d}{eK}.
  \end{equation}
  By the definition of $T_{\mathcal{D}}$, we have $T_{\mathcal{D}}\sim Binomial\left(m,p_d\right)$. Applying Lemma~\ref{lemma:chernoff}, 
  \begin{equation}
    \mathbb{P}\left(T_{\mathcal{D}}<\frac{mp_d}{2}\right)\leq \exp\left(-\frac{mp_d}{8}\right).
  \end{equation}
  Applying the union bound over all subsets $\mathcal{D}\subseteq \mathcal{I}$, yields
  \begin{equation}
    \begin{aligned}
      & \mathbb{P}\left(\text{there exists at least one subset}\ \mathcal{D}\ \text{with}\ T_{\mathcal{D}} <{mp_{d}}/{2}\right) \\
      &\leq \sum_{d=1}^{K}{\binom{K}{d}\exp\left(-\frac{mp_d}{8}\right)} \leq \sum_{d=1}^K \exp\left(d\log K -\frac{md}{8eK}\right),
    \end{aligned}
  \end{equation}
  which is upper bounded by $\delta/4$ by the stated condition on $m$. This completes the proof of Lemma~\ref{lemma:test coverage}.
\end{proof}

\begin{lemma}\label{lemma:candidate eliminate}
  Condition on the high probability events in Lemmas~\ref{lemma:candidate set properties} and~\ref{lemma:test coverage}. With probability at least $1-\delta/4$, for every non-empty $\mathcal{D}\subseteq \mathcal{I}$ and every $\mathcal{A}\subseteq \mathcal{R}$ with $\left|\mathcal{A}\right|\leq \lambda \left|\mathcal{D}\right|$, there exists at least one pool that contains at least one sample from $\mathcal{D}$ and no sample from $\left(\mathcal{I}\backslash\mathcal{D}\right) \cup \mathcal{A}$, provided that
  \begin{equation}
    m\geq 2e^{\beta+1}K\left(\log\left({4K^2}/{\delta}\right) + \lambda \log\left(\beta K+1\right)\right).
  \end{equation} 
\end{lemma}

\begin{proof}[Proof of Lemma~\ref{lemma:candidate eliminate}]
  Given a non-empty subset $\mathcal{D}\subseteq \mathcal{I}$ with $\left|\mathcal{D}\right|=d$, and a subset $\mathcal{A} \subseteq \mathcal{R}$ with $\left|\mathcal{A}\right|=a \leq \lambda d$. By Lemma~\ref{lemma:test coverage}, there are at least $T_{\mathcal{D}}\geq mp_d /2$ pools that contain samples from $\mathcal{D}$ and no sample from $\mathcal{I}\backslash\mathcal{D}$. For any such pool, the probability that it contains no sample from $\mathcal{A}$ is $\left(1-p\right)^a$, which satisfies
  \begin{equation}
    (1-p)^a \geq (1-p)^{\beta K} \geq e^{-\beta},
  \end{equation}
  where $a\leq \left|\mathcal{R}\right|\leq \beta K$. Define $E_{\mathcal{D},\mathcal{A}}$ as the event that each of these $T_{\mathcal{D}}$ pools contains at least one sample from $\mathcal{A}$, which occurs with probability
  \begin{equation}
    \left(1-(1-p)^a\right)^{T_{\mathcal{D}}} \leq \exp\left(-\frac{md}{2e^{\beta+1}K}\right).
  \end{equation}

  We next apply the union bound over all possible pairs $\left(\mathcal{D},\mathcal{A}\right)$. For a fixed $d$, the number of pairs is at most
  \begin{equation}
    \sum_{a=0}^{\min\left\{\lambda d, \left|\mathcal{R}\right|\right\}} {\binom{\left|\mathcal{R}\right|}{a}\binom{K}{d}} \leq K^d \sum_{a=0}^{\lambda d} {\binom{\beta K}{a}} \leq K^d \left(\beta K + 1\right)^{\lambda d}.
  \end{equation}
  Therefore, conditioned on the high probability events in Lemmas~\ref{lemma:candidate set properties} and~\ref{lemma:test coverage}, the probability that there exists a pair $\left(\mathcal{D},\mathcal{A}\right)$ for which $E_{\mathcal{D},\mathcal{A}}$ occurs is upper bounded by
  \begin{equation}
    \sum_{d=1}^{K}{K^d \left(\beta K + 1\right)^{\lambda d} \exp\left(-\frac{md}{2e^{\beta+1}K}\right)} \leq \frac{\delta}{4},
  \end{equation}
  where the last inequality follows from the stated condition on $m$, thus completing the proof of Lemma~\ref{lemma:candidate eliminate}.
\end{proof}

\begin{proof}[Proof of Theorem~\ref{thm:sufficient condition for noiseless recovery}]
  By the union bound, the events established in Lemmas~\ref{lemma:candidate set properties},~\ref{lemma:test coverage} and~\ref{lemma:candidate eliminate} occur simultaneously with probability at least $1-\delta$, under which $\mathcal{I}$ is the unique minimizer of (\ref{eq:recovery_problem_noiseless_set}).
  
  Define the objective function of the problem in~(\ref{eq:recovery_problem_noiseless_set}) as 
  \begin{equation}
    \Phi\left(\mathcal{S}\right)=\gamma\left|\mathcal{S}\right|+\left|\partial_{\mathcal{E}} \mathcal{S}\right|.
  \end{equation}
  Consider any $\mathcal{S}\subseteq\mathcal{U}$ with $\mathcal{S}\neq \mathcal{I}$, which can be expressed as $\mathcal{S}=\left(\mathcal{I} \backslash \mathcal{D}\right)\cup \mathcal{A}$, where $\mathcal{A}\subseteq \mathcal{R}$ and $\mathcal{D}\subseteq \mathcal{I}$. Thus, we have
  $\mathcal{D}=\mathcal{I}\backslash\mathcal{S}$ and $\mathcal{A}=\mathcal{S}\backslash\mathcal{I}$.
  
  In the case of $\mathcal{D}=\varnothing$, $\mathcal{S}=\mathcal{I} \cup \mathcal{A}$ with $\mathcal{A}\neq \varnothing$. By Lemma~\ref{lemma:candidate set properties}, there are no edges between $\mathcal{A}$ and $\mathcal{I}$, which yields 
  \begin{equation}
    \left|\partial_{\mathcal{E}} \mathcal{S}\right| = \left|\partial_{\mathcal{E}} \mathcal{I}\right| + \left|\partial_{\mathcal{E}} \mathcal{A}\right|.
  \end{equation}
  Thus, 
  \begin{equation}
    \Phi\left(\mathcal{S}\right) - \Phi\left(\mathcal{I}\right) = \gamma \left|\mathcal{A}\right| + \left|\partial_{\mathcal{E}} \mathcal{A}\right| > 0,
  \end{equation}
  implying that $\mathcal{S}$ is strictly suboptimal.

  In the case of $\mathcal{D}\neq\varnothing$, we have 
  \begin{equation}
    \left|\partial_{\mathcal{E}} \mathcal{S}\right| = \left|\partial_{\mathcal{E}} \left({\mathcal{I}\backslash \mathcal{D}}\right)\right| + \left|\partial_{\mathcal{E}} \mathcal{A}\right|.
  \end{equation}
  Hence, 
  \begin{equation}\label{eq:objective_diff}
    \begin{aligned}
      \Phi\left(\mathcal{S}\right) - \Phi\left(\mathcal{I}\right) = \left|\mathcal{E}\left(\mathcal{D},\mathcal{I}\backslash\mathcal{D}\right)\right| - \left|\mathcal{E}\left(\mathcal{D},\mathcal{V}\backslash\mathcal{I}\right)\right| 
      + \left|\partial_{\mathcal{E}} \mathcal{A}\right| + \gamma \left(\left|\mathcal{A}\right|-\left|\mathcal{D}\right|\right).
    \end{aligned} 
  \end{equation}
  Since $\left|\mathcal{E}\left(\mathcal{D},\mathcal{V}\backslash\mathcal{I}\right)\right|\leq d_{max}\left|\mathcal{D}\right|$, the term in~(\ref{eq:objective_diff}) is lower bounded by
  \begin{equation}
    \Phi\left(\mathcal{S}\right) - \Phi\left(\mathcal{I}\right) \geq \gamma \left(\left|\mathcal{A}\right|-\left|\mathcal{D}\right|\right) - d_{max}\left|\mathcal{D}\right|.
  \end{equation}

  If $\Phi\left(\mathcal{S}\right) > \Phi\left(\mathcal{I}\right)$, then $\mathcal{S}$ is strictly suboptimal. Otherwise, suppose that $\Phi\left(\mathcal{S}\right) \leq \Phi\left(\mathcal{I}\right)$. Then,
  \begin{equation}
    \left|\mathcal{A}\right| \leq \left(1+\frac{d_{max}}{\gamma}\right) \left|\mathcal{D}\right| = \lambda \left|\mathcal{D}\right|.
  \end{equation}
  By Lemma~\ref{lemma:candidate eliminate}, there exists a pool that contains at least one sample from $\mathcal{D}$ and no sample from $\left(\mathcal{I}\backslash\mathcal{D}\right) \cup \mathcal{A}$. This pool is positive but cannot be explained by $\mathcal{S}$. Consequently, $\mathcal{S}$ violates the constraints in~(\ref{eq:recovery_problem_noiseless_set}). 

  In both cases, every $\mathcal{S}\neq \mathcal{I}$ is either infeasible or strictly suboptimal. Thus, $\mathcal{I}$ is the unique optimal solution to~(\ref{eq:recovery_problem_noiseless_set}), completing the proof of Theorem~\ref{thm:sufficient condition for noiseless recovery}.
\end{proof}

\section{Proof of  Theorem~\ref{thm:sufficient condition for noisy recovery}}\label{proof:sufficient condition for noisy recovery}
To prove Theorem~\ref{thm:sufficient condition for noisy recovery}, we need the Bernstein's inequality.
\begin{lemma}[Bernstein's inequality~\cite{9780199535255.001.0001, beygelzimer2016search}]\label{lemma:bernstein}
  Let $X_1,\ldots,X_n$ be independent zero-mean random variables. Suppose that $\left|X_i\right|\leq B$ almost surely. Then, for all $t>0$, 
  \begin{equation}
    \mathbb{P}\left(\sum_{i=1}^{n}{X_i} \geq t\right) \leq \exp\left(-\frac{t^2/2}{\sum_{i=1}^{n} {\mathbb{E}\left[X_i^2\right]} + Bt/3}\right).
  \end{equation}
\end{lemma}

Similar to the noiseless setting, we first reformulate the problem in~(\ref{eq:recovery_problem_noisy}) with respect to the support of the binary vector $z$. For any $\mathcal{S}\subseteq\mathcal{V}$, define 
\begin{equation}
  \phi_i\left(\mathcal{S}\right)=\mathbb{I}\left(\sum_{j\in\mathcal{S}}{M_{i,j}}\geq 1\right),\ h_i(\mathcal{S})=\mathbb{P}\left(\phi_i\left(\mathcal{S}\right)\neq\phi_i\left(\mathcal{I}\right)\right),\ i=1,\ldots,m,
\end{equation}
and
\begin{equation}
  \begin{aligned}
    H\left(\mathcal{S}\right) = \frac{1}{m}\sum_{i=1}^{m}{\mathbb{I}\left(\tilde{y}_i \neq \phi_i\left(\mathcal{S}\right)\right)}.
  \end{aligned}
\end{equation}
Since $h_i(\mathcal{S})$ is independent of $i$, denote the common value by $h\left(\mathcal{S}\right)$. For any $z$ with $\mathcal{S}=\operatorname{supp}\left(z\right)$, we have $d_H\left(\tilde{y},\hat{y}\left(z\right)\right)=H\left(\mathcal{S}\right)$.
Thus, the problem in~(\ref{eq:recovery_problem_noisy}) is reformulated as
\begin{equation}\label{eq:recovery_problem_noisy_set}
  \min_{\mathcal{S}\subseteq \mathcal{U_{\epsilon}}} \gamma\left|\mathcal{S}\right| + \left|\partial_{\mathcal{E}} \mathcal{S}\right| + \mu H\left(\mathcal{S}\right).
\end{equation}

Theorem~\ref{thm:sufficient condition for noisy recovery} follows from the following two lemmas.

\begin{lemma}\label{lemma:candidate set properties noisy}
  Consider the system model in~(\ref{eq:system_noisy}) with flip probability $\epsilon\in\left(0,1/2\right)$, where $x$ is a $K$-sparse signal with support $\mathcal{I}\in \mathcal{F}_{K,k}$, and $M$ is a Bernoulli random matrix with $p=1/\left(1+K\right)$. Define $\lambda = 1+d_{max}/\gamma$. Let $\beta\geq \lambda$ and $\mathcal{R}_{\epsilon}=\mathcal{U_{\epsilon}}\backslash\mathcal{I}$. Each of the following statements holds with probability at least $1-\delta/3$:
  \begin{itemize}
    \item [(i)] $\left|\mathcal{U}_{\epsilon}\right|\leq \left(1+\beta\right)K$, provided that
    \begin{equation}
      m \geq \frac{8e\left(1+(e-2)\epsilon\right)(K+1)}{(1-2\epsilon)^2} \log\left({3(n-K)}/{\delta \beta K}\right).
    \end{equation}
    \item [(ii)] $\mathcal{I}\subseteq \mathcal{U}_{\epsilon}$, provided that
    \begin{equation}
      m\geq \frac{4e\left(1+(6e-2)\epsilon\right)(K+1)}{3(1-2\epsilon)^2} \log\left({3K}/{\delta}\right).
    \end{equation}
  \end{itemize}
\end{lemma}

\begin{proof}[Proof of Lemma~\ref{lemma:candidate set properties noisy}]
  Recall that the candidate set is given by 
  \begin{equation}
    \mathcal{U}_{\epsilon}=\left\{j\in\mathcal{V}\mid \tilde{N}_j <\tau \right\},
  \end{equation}
  where $\tilde{N}_j$ is the number of observed negative pools containing sample $j$ and $\tau$ is a threshold. For a sample $j\notin \mathcal{I}$, define $q_{\epsilon}$ as the probability that a pool containing the sample $j$ is observed as negative, which satisfies
  \begin{equation}
    q_\epsilon=p[(1-p)^K(1-\epsilon)+(1-(1-p)^K)\epsilon] \geq \frac{\alpha(\epsilon)}{K+1},
  \end{equation}
  where $\alpha\left(\epsilon\right)=\left({1+(e-2)\epsilon}\right)/{e}$. Hence, we have $\tilde{N}_j\sim Binomial\left(m,q_{\epsilon}\right)$ for each sample $j\notin \mathcal{I}$. The threshold is set to 
  \begin{equation}
    \tau=\frac{m\left(\alpha\left(\epsilon\right)+\epsilon\right)}{2(K+1)}.
  \end{equation}
  Applying Lemma~\ref{lemma:chernoff}, yields
  \begin{equation}
    \mathbb{P}\left(\tilde{N}_j < \tau\right) \leq \exp\left(-\frac{m\left(\alpha(\epsilon)-\epsilon\right)^2}{8\alpha(\epsilon)(K+1)}\right).
  \end{equation}
  Let $\mathcal{R}_{\epsilon}=\mathcal{U}_{\epsilon}\backslash \mathcal{I}$ be the set of negative samples in $\mathcal{U}_{\epsilon}$. Then,
  \begin{equation}
    \mathbb{E}\left[\left|\mathcal{R}_{\epsilon}\right|\right] \leq \left(n-K\right)\exp\left(-\frac{m\left(\alpha(\epsilon)-\epsilon\right)^2}{8\alpha(\epsilon)(K+1)}\right).
  \end{equation}
  By Markov's inequality, we have
  \begin{equation}
    \mathbb{P}\left(\left|\mathcal{R}_{\epsilon}\right|>\beta K\right)\leq \frac{\mathbb{E}\left[\left|\mathcal{R}_{\epsilon}\right|\right]}{\beta K},
  \end{equation} 
  which is bounded by $\delta/3$ provided that
  \begin{equation}
    m \geq \frac{8e\left(1+(e-2)\epsilon\right)(K+1)}{(1-2\epsilon)^2} \log\left({3(n-K)}/{\delta \beta K}\right).
  \end{equation}
  This proves the part (i).

  We next prove the part (ii). For any sample $j\in\mathcal{I}$, a pool containing $j$ is observed as negative only if the test outcome is flipped. Hence, we have $\tilde{N}_j\sim Binomial\left(m,p\epsilon\right)$ for each sample $j\in \mathcal{I}$. Applying Lemma~\ref{lemma:chernoff}, yields
  \begin{equation}
    \mathbb{P}\left(\tilde{N}_j \geq \tau\right) \leq \exp\left(-\frac{3m\left(\alpha(\epsilon)-\epsilon\right)^2}{4\left(\alpha(\epsilon) + 5\epsilon\right)(K+1)}\right).
  \end{equation}
  By the union bound over all $j\in\mathcal{I}$, yields
  \begin{equation}
    \mathbb{P}\left(\mathcal{I}\not\subseteq\mathcal{U}_{\epsilon}\right) \leq K \exp\left(-\frac{3m\left(\alpha(\epsilon)-\epsilon\right)^2}{4\left(\alpha(\epsilon) + 5\epsilon\right)(K+1)}\right),
  \end{equation}
  which is at most $\delta/3$ provided that
  \begin{equation}
    m\geq \frac{4e\left(1+(6e-2)\epsilon\right)(K+1)}{3(1-2\epsilon)^2} \log\left({3K}/{\delta}\right).
  \end{equation}
  This completes the proof of Lemma~\ref{lemma:candidate set properties noisy}.
\end{proof}

\begin{lemma}\label{lemma:hamming distance}
  Under the same setting as Lemma~\ref{lemma:candidate set properties noisy}, consider any $\mathcal{S}=\left(\mathcal{I} \backslash \mathcal{D}\right)\cup \mathcal{A}$ with $\mathcal{S}\neq\mathcal{I}$, where $\mathcal{D}\subseteq\mathcal{I}$ and $\mathcal{A}\subseteq\mathcal{V}\backslash\mathcal{I}$ with $\left|\mathcal{A}\right|\leq \beta K$. Then, with probability at least $1-\delta/3$, 
  \begin{equation}
    H\left(\mathcal{S}\right)- H\left(\mathcal{I}\right)>\frac{{(1-2\epsilon)}h(\mathcal{S})}{2}
  \end{equation}
  holds for all such $\mathcal{S}$, provided that
  \begin{equation}
    m\geq \frac{16e^{\beta+1}(2-\epsilon)K}{3(1-2\epsilon)^2}\log\left({n}/{\log\left(1+\delta/3\right)}\right).
  \end{equation}
\end{lemma}

\begin{proof}[Proof of Lemma~\ref{lemma:hamming distance}]
  Given a $\mathcal{S}=\left(\mathcal{I} \backslash \mathcal{D}\right)\cup \mathcal{A}$. Let $\left|\mathcal{D}\right|=d$, $\left|\mathcal{A}\right|=a\leq \beta K$, and $r=\left|\mathcal{S}\Delta\mathcal{I}\right|$, where $\mathcal{S}\Delta\mathcal{I}$ denotes the symmetric difference between $\mathcal{S}$ and $\mathcal{I}$. By the definition of $h\left(\mathcal{S}\right)$, we have
  \begin{equation}
    \begin{aligned}
      h\left(\mathcal{S}\right) & =  \left(1-(1-p)^d\right)(1-p)^{K-d+a} 
      + (1-p)^{K}\left(1-(1-p)^a\right) \\
      & \geq (a+d)p(1-p)^{K+a-1} \geq \frac{r}{e^{\beta+1}K}.
    \end{aligned} 
  \end{equation}
  For each pool $i$, define
  \begin{equation}
    W_i=\mathbb{I}\left(\tilde{y}_i\neq \phi_i\left(\mathcal{S}\right)\right) - \mathbb{I}\left(\tilde{y}_i\neq \phi_i\left(\mathcal{I}\right)\right),\ i=1,\ldots,m.
  \end{equation}
  Then, we have
  \begin{equation}
    H\left(\mathcal{S}\right) - H\left(\mathcal{I}\right) = \frac{1}{m}\sum_{i=1}^{m}{W_i}.
  \end{equation}
  If $\phi_i\left(\mathcal{S}\right) = \phi_i \left(\mathcal{I}\right)$, then $W_i=0$. Otherwise, $W_i=1$ when the test outcome is not flipped and $W_i=-1$ when the test outcome is flipped. Thus, for $i=1,\ldots,m$,
  \begin{equation}
      \mathbb{E}\left[W_i\right]=(1-2\epsilon)h(\mathcal{S}), \
      \mathbb{E}\left[W_i^2\right]=h\left(\mathcal{S}\right).
  \end{equation}
  Define $X_i=W_i-\mathbb{E}\left[W_i\right]$ for $i=1,\ldots,m$, which satisfies 
  \begin{equation}
    \left|X_i\right|\leq 2,\ \mathbb{E}\left[X_i^2\right] \leq h\left(\mathcal{S}\right).
  \end{equation}
  Hence,
  \begin{equation}
    \begin{aligned}
      & \mathbb{P}\left(H\left(\mathcal{S}\right) - H\left(\mathcal{I}\right) \leq \frac{{(1-2\epsilon)}h(\mathcal{S})}{2}\right) \\
      & = \mathbb{P}\left(\sum_{i=1}^{m}{X_i} \leq - \frac{m(1-2\epsilon)h\left(\mathcal{S}\right)}{2}\right) \\
      & \leq \exp\left(-\frac{3m(1-2\epsilon)^2 r}{16e^{\beta+1}(2-\epsilon)K}\right).
    \end{aligned}
  \end{equation}

  Let
  \begin{equation}
    \eta=\exp\left(-\frac{3m(1-2\epsilon)^2}{16e^{\beta+1}(2-\epsilon)K}\right).
  \end{equation}
  Applying the union bound over all possible pairs $\left(\mathcal{D},\mathcal{A}\right)$, the total failure probability is at most
  \begin{equation}
    \sum_{\substack{0\leq d \leq K \\ 0\leq a \leq \beta K \\ a+d\geq 1}} \binom{K}{d}\binom{n-K}{a} \eta^{d+a} \leq (1+\eta)^n - 1 \leq e^{n\eta}-1,
  \end{equation}
  which is bounded by $\delta/3$ by the stated condition on $m$. 
\end{proof}

\begin{proof}[Proof of Theorem~\ref{thm:sufficient condition for noisy recovery}]
  By the union bound, the events established in Lemmas~\ref{lemma:candidate set properties noisy} and~\ref{lemma:hamming distance} occur simultaneously with probability at least $1-\delta$, under which $\mathcal{I}$ is the unique minimizer of~(\ref{eq:recovery_problem_noisy_set}).

  Define the objective function of the problem in~(\ref{eq:recovery_problem_noisy_set}) as
  \begin{equation}\label{eq:objective_noisy}
    \Psi\left(\mathcal{S}\right)=\gamma\left|\mathcal{S}\right| + \left|\partial_{\mathcal{E}} \mathcal{S}\right| + \mu H\left(\mathcal{S}\right).
  \end{equation}
  Consider any $\mathcal{S}\subseteq\mathcal{U}_{\epsilon}$ with $\mathcal{S}\neq \mathcal{I}$, which can be expressed as $\mathcal{S}=\left(\mathcal{I} \backslash \mathcal{D}\right)\cup \mathcal{A}$, where $\mathcal{A}\subseteq\mathcal{U}_{\epsilon}\backslash\mathcal{I}$ and $\mathcal{D}\subseteq\mathcal{I}$. Then, we have
  \begin{equation}\label{eq:edge_diff}
    \left|\partial_{\mathcal{E}} \mathcal{S}\right| - \left|\partial_{\mathcal{E}} \mathcal{I}\right| \geq -d_{max}r,\ \left|\mathcal{S}\right| - \left|\mathcal{I}\right| \geq -r.
  \end{equation}
  Since $\left|\mathcal{A}\right|\leq \beta K$ by Lemma~\ref{lemma:candidate set properties noisy}, Lemma~\ref{lemma:hamming distance} yields
  \begin{equation}\label{eq:H_diff}
    H\left(\mathcal{S}\right)- H\left(\mathcal{I}\right)>\frac{{(1-2\epsilon)}h(\mathcal{S})}{2} \geq \frac{(1-2\epsilon)r}{2e^{\beta+1}K}.
  \end{equation}

  Combining~(\ref{eq:edge_diff}) and~(\ref{eq:H_diff}), the term in~(\ref{eq:objective_noisy}) is lower bounded by
  \begin{equation}
    \Psi\left(\mathcal{S}\right) - \Psi\left(\mathcal{I}\right) > \frac{\mu(1-2\epsilon)r}{2e^{\beta+1}K} - (d_{max}+\gamma)r.
  \end{equation}
  When 
  \begin{equation}
    \mu \geq \frac{2e^{\beta+1}K(d_{max}+\gamma)}{1-2\epsilon},
  \end{equation}
  $\Psi\left(\mathcal{S}\right) > \Psi\left(\mathcal{I}\right)$ holds for every $\mathcal{S}\subseteq\mathcal{U}_{\epsilon}$ with $\mathcal{S}\neq \mathcal{I}$. Thus, $\mathcal{I}$ is the unique minimizer of~(\ref{eq:recovery_problem_noisy_set}), completing the proof of Theorem~\ref{thm:sufficient condition for noisy recovery}.
\end{proof}
}

\bibliographystyle{IEEEtran}
\bibliography{IEEEabrv,reference}

\end{document}